\documentclass[draftcls, onecolumn, letterpaper, romanappendices]{IEEEtran}

\usepackage{graphics}
\usepackage{cite}
\usepackage[pdftex]{graphicx}
\usepackage{epstopdf}
\usepackage{epsfig}
\usepackage{latexsym}
\usepackage{amsfonts}
\usepackage{calc}
\usepackage{url}
\usepackage{enumerate}
\usepackage{color}
\usepackage[tbtags]{amsmath}
\usepackage{amssymb}
\usepackage{upref}
\usepackage{dsfont}
\usepackage{multirow}
\usepackage{booktabs}
\usepackage{bigstrut}
\usepackage{rotating}
\usepackage{nth}
\usepackage{breqn}

\newtheorem{theorem}{Theorem}

\newtheorem{lemma}{Lemma}

\newtheorem{proof}[theorem]{Proof}

\renewcommand{\IEEEQED}{\IEEEQEDopen}

\begin{document}
\title{Coding Delay Analysis of Dense and Chunked Network Codes over Line Networks${}^{\dagger}$ \footnote{${}^{\dagger}$A preliminary version of this work has been presented partly in NetCod 2012, Cambridge, MA, USA, June 2012, and in part in ISIT 2012, Cambridge, MA, USA, July 2012.}}

\author{\IEEEauthorblockN{Anoosheh~Heidarzadeh and Amir H. Banihashemi\\}
\IEEEauthorblockA{\small{Department of Systems and Computer Engineering, Carleton University, Ottawa, ON, Canada}\\
}}

\maketitle

\begin{abstract}In this paper, we analyze the coding delay and the average coding delay of random linear network codes (a.k.a. dense codes) and chunked codes (CC), which are an attractive alternative to dense codes due to their lower complexity, over line networks with Bernoulli losses and deterministic regular or Poisson transmissions. Our results, which include upper bounds on the delay and the average delay, are (i) for dense codes, in some cases more general, and in some other cases tighter, than the existing bounds, and provide a more clear picture of the speed of convergence of dense codes to the (min-cut) capacity of line networks; and (ii) the first of their kind for CC over networks with such probabilistic traffics. In particular, these results demonstrate that a stand-alone CC or a precoded CC provide a better tradeoff between the computational complexity and the convergence speed to the network capacity over the probabilistic traffics compared to arbitrary deterministic traffics which have previously been studied in the literature.\end{abstract}

\section{Introduction}
Random linear network codes (a.k.a. dense codes) achieve the min-cut capacity over various network scenarios, e.g., unicast over line networks, but at the cost of a rather high computational complexity~\cite{LMKE:2008}. Targeting the design of more computationally efficient network codes, Maymounkov \emph{et al.}~\cite{MHL:2006} proposed \emph{chunked codes} (CC), which generalize dense codes, and operate by partitioning the message of the source into non-overlapping (disjoint) sub-messages of equal size, called \emph{chunks}~\cite{MHL:2006}. Recently, a generalized version of chunked codes, referred to as \emph{overlapped chunked codes} (OCC), were also independently proposed in \cite{SZK:2009} and \cite{HB:2010}. It has been analytically shown in \cite{HB:2011} that, for sufficiently small chunks, OCC provide a better tradeoff between the speed of convergence to (achieve or approach) the min-cut capacity and the message or packet error rate, compared to CC, over line networks with arbitrary deterministic traffics. This is while earlier in \cite{HB:2010} it was demonstrated that CC provide a better tradeoff between the speed of convergence to the min-cut capacity and the message error rate for sufficiently large chunks (also see~\cite{HBJ:2011} for more details). In this paper, our focus is on chunked codes. The extension of the analysis to OCC is not straightforward and is beyond the scope of this work. In chunked coding, each node at each transmission time randomly chooses a chunk, and transmits it by using a dense code. In fact, a dense code is a CC with only one chunk of the message size. Thus, CC require less complex coding operations due to applying coding on chunks smaller than the original message. This however comes at the cost of lower speed of convergence to the min-cut capacity compared to dense codes.

The speed of convergence of dense codes and chunked codes to the min-cut capacity of line networks with arbitrary deterministic traffics (with deterministic transmission schedules and loss models) was studied in~\cite{MHL:2006,HBJ:2011}. It is not however straightforward to apply the results to the networks with probabilistic traffics. In particular, it has been shown that for arbitrary deterministic traffics (i) a dense code always achieves the capacity; (ii) a CC achieves the min-cut capacity, so long as the size of the chunks is lower bounded by a function super-logarithmic in the message size and super-log-cubic in the network length, and (iii) a CC, preceded by a capacity-achieving erasure code, approaches the min-cut capacity with an arbitrarily small but non-zero constant gap, so long as the size of the chunks is lower bounded by a function constant in the message size and log-cubic in the network length.

Aside from the results for arbitrary deterministic traffics, Lun \emph{et al.} \cite{LMKE:2008} showed that dense codes achieve the min-cut capacity of line networks with probabilistic traffics specified by stochastic processes with bounded average rate. They however did not discuss the speed of convergence of such codes to the min-cut capacity. This issue was later studied in~\cite{PFS:2005,DDHE:2009}, by analyzing the coding delay\footnote{The \emph{coding delay} of a code over a network with a given traffic (schedule of transmissions and losses) is the minimum time that the code takes to transmit all the message vectors from the source to the sink. The coding delay is a random variable due to the randomness in both the code and the traffic.} and the average coding delay\footnote{The \emph{average coding delay} of a code with respect to a class of traffics is the coding delay of the code averaged out over all the traffics (but not the codes), and hence is a random variable due to the randomness in the code.} of dense codes over some probabilistic traffics. There is however no result on CC over the networks with probabilistic traffics in the literature.

Pakzad \emph{et al.} \cite{PFS:2005}, for the first time, studied the average coding delay of dense codes (operating in $\mathbb{F}_2$) over line networks with deterministic regular transmissions and Bernoulli losses, where the special case of two identical links in tandem was considered. The analysis however did not provide any insight about how the coding delay (which is random with respect to both the codes and the traffics) can deviate from the average coding delay (which is random with respect to the codes but not the traffics).

More recently, Dikaliotis \emph{et al.} \cite{DDHE:2009} studied both the average coding delay and the coding delay of dense codes (operating in a finite field of infinitely large size) over the line networks of arbitrary length with traffics similar to those in \cite{PFS:2005}, but under the assumption that there exists a unique worst link (i.e., a unique link with the minimum probability of transmission success) in the network.

In this paper, we study the coding delay and the average coding delay of dense codes, and for the first time, chunked codes for different ranges of the chunk sizes, operating in the field of size two ($\mathbb{F}_2$), over line networks with traffics similar to those in \cite{LMKE:2008,PFS:2005,DDHE:2009}. Our study has no limiting assumption on the traffic parameters or the length of the network. It is worth noting that any upper bound on the coding delay or on the average coding delay of any coding scheme over $\mathbb{F}_2$ serves as an upper bound for the underlying code over any finite field of larger size. The method of analysis in this paper is itself, however, generalizable to finite fields of larger size, but the generalization is not trivial and is beyond the scope of this paper.

The main contributions of this paper are:

\begin{itemize}
\item We derive upper bounds on the coding delay and the average coding delay of a dense code, or a CC alone, or a CC with precoding, in the asymptotic setting, i.e., as the message size tends to infinity, over the traffics with deterministic regular transmissions or Poisson transmissions and Bernoulli losses with arbitrary parameters.\footnote{The scenario of deterministic regular transmissions and Bernoulli losses has been studied in \cite{PFS:2005,DDHE:2009}, and the scenario of Poisson transmissions with Bernoulli losses has been studied in \cite{LMKE:2008} as a special case of the probabilistic traffics over line networks.} The upper bounds are functions of the message size, the length of the network, and the parameters of the traffic and the code. We also consider a special case with unequal traffic parameters, where no two links have equal traffic parameters. The upper bounds, in this case, indicate how the coding delay or the average coding delay change as a function of the minimum of the (absolute value of the) difference between the traffic parameters of any two consecutive links in the network.
\item We show that: (i) our upper bounds on the average coding delay of dense codes are in some cases more general, and in some other cases tighter, than what were presented in \cite{PFS:2005,DDHE:2009}, and (ii) the coding delay of dense codes may have a large deviation from the average coding delay in both cases of identical and non-identical links; for non-identical links, our upper bound on such a deviation is smaller than what was previously shown in \cite{DDHE:2009}. It is noteworthy that, for identical links, upper bounding such a deviation has been an open problem (see \cite{DDHE:2009}).
\item We also show that: (i) a CC achieves the min-cut capacity, so long as the size of the chunks is bounded from below by a function super-logarithmic in the message size and super-log-linear in the network length, and (ii) the combination of a CC and a capacity-achieving erasure code approaches the min-cut capacity with an arbitrarily small non-zero constant gap, so long as the size of the chunks is bounded from below by a function constant in the message size and log-linear in the network length. The lower bounds in both cases are smaller than those over the networks with arbitrary deterministic traffics. Thus both coding schemes (i.e., stand-alone CC and CC with precoding) are less computationally complex (require smaller chunks), for the same speed of convergence (with or without a gap) to the min-cut capacity, over such probabilistic traffics, compared to arbitrary deterministic traffics.
\item In a capacity-achieving scenario, for such probabilistic traffics, we show that for CC: (i) the upper bound on the overhead (the difference between the coding delay and the min-cut capacity\footnote{For the definition of min-cut capacity used here, see Section~\ref{subsec:TLM}}) grows sub-log-linearly with the message size and the network length, and decays sub-linearly with the size of the chunks, and (ii) the upper bound on the average overhead (the difference between the average coding delay and the min-cut capacity) grows sub-log-linearly (or poly-log-linearly) with the message size, and sub-log-linearly (or log-linearly) with the network length, and decays sub-linearly (or linearly) with the size of the chunks, in the case with arbitrary (or unequal) traffic parameters. For arbitrary deterministic traffics, the upper bound on the overhead or that on the average overhead was shown in~\cite{HBJ:2011} to be similar to the case (i), mentioned above, but with a larger (super-linear) growth rate with the network length.
\end{itemize}

This paper is an extended version of our earlier works \cite{HB1P:2012,HB2P:2012}, and contains more details on the methodology of the analysis and the details of the proofs.



\section{Network Model and Problem Setup}
\subsection{Transmission and Loss Model}\label{subsec:TLM}
We consider a unicast problem (one-source one-sink) over a line network with $L$ links connecting $L+1$ nodes $\{v_i\}_{0\leq i\leq L}$ in tandem. The source node $v_0$ has a message of $k$ vectors, called \emph{message vectors}, from a vector space $\mathcal{F}$ over $\mathbb{F}_2$, and the sink node $v_L$ requires all the message vectors.

Each (non-sink) node at each transmission time transmits a (coded) packet, which is a vector in $\mathcal{F}$. The packet transmissions are assumed to occur in discrete-time, and the transmission times over different links are assumed to follow independent stochastic processes. The transmission times over the $i\textsuperscript{th}$ link are specified by (i) a deterministic process where there is a packet transmission at each time instant, or (ii) a Poisson process with parameter $\lambda_i: 0<\lambda_i\leq 1$, where $\lambda_i$ is the average number of transmissions per time unit over the $i\textsuperscript{th}$ link. The transmission schedules resulting from (i) and (ii) are referred to as \emph{deterministic regular} and \emph{Poisson}, respectively.

Each transmitted packet either succeeds (\emph{successful packet}) or fails (\emph{lost packet}) to be received. The successful packets are assumed to arrive with zero delay, and the lost packets will never arrive. The packets are assumed to be successful independently over different links. The successful packets over the $i\textsuperscript{th}$ link are specified by a Bernoulli process with (success) parameter $p_i: 0<p_i\leq 1$, where $p_i$ is the average number of successes per transmission over the $i\textsuperscript{th}$ link. The loss model defined as above is referred to as \emph{Bernoulli}.

For each traffic model as above, the parameters $\{p_i\}$ or $\{\lambda_i,p_i\}$ are called the \emph{traffic parameters}. In the case of traffics with parameters $\{p_i\}$ or $\{\lambda_i,p_i\}$, the \emph{min-cut capacity} is defined as the ratio of the message size $k$ to the minimum (equivalent) traffic parameter $\min_{1\leq i\leq L}p_i$ or $\min_{1\leq i\leq L}\lambda_i p_i$, respectively. For simplifying the terminology, hereafter, we refer to the ``min-cut capacity'' as the ``capacity.''

\subsection{Assumptions}
We assume that there is no feedback information in the network before the time that the sink node recovers all the message packets. Whenever the decoding process is successful, the sink node $v_L$ sends an acknowledge message to the node $v_{L-1}$. The node $v_{L-1}$ then stops transmitting new packets to the sink node, and relays the acknowledge message to the node $v_{L-2}$. The feedback relaying process continues over the links till the time that the source node $v_0$ receives the acknowledge message, and stops transmitting new packets. The feedback transmissions are assumed to be error-free and with zero delay.

We also assume that the size of the memory at the network nodes is unbounded, i.e., all the packets, received by a node, will remain in the memory of that node till the end of the transmission time.

\subsection{Problem Setup}
The goal in this paper is to derive upper bounds on the coding delay and the average coding delay of dense codes and chunked codes over line networks with deterministic regular or Poisson transmissions and Bernoulli losses.

For some fixed $0<\epsilon<1$, the coding delay of a class of codes over a network with a class of traffics is upper bounded by $N$ with failure probability (w.f.p.) bounded above by (b.a.b.) $\epsilon$, so long as the coding delay of a randomly chosen code over the network with a randomly chosen traffic is larger than $N$ with probability (w.p.) b.a.b. $\epsilon$. The average coding delay of a class of codes over a network with respect to a class of traffics is upper bounded by $N$ w.f.p. b.a.b. $\epsilon$, so long as the average coding delay of a randomly chosen code over the network with respect to the class of traffics is larger than $N$ w.p. b.a.b. $\epsilon$.

\subsection{Asymptotic Notations}
Throughout the paper, we will use the asymptotic notations $O(.)$, $o(.)$, $\Omega(.)$ and $\omega(.)$ defined as follows. For non-negative functions $f(n)$ and $g(n)$, we write: (i) $f(n)=O(g(n))$ if and only if $\limsup_{n\rightarrow\infty}\frac{f(n)}{g(n)}<\infty$; (ii) $f(n)=o(g(n))$ if and only if $\lim_{n\rightarrow\infty}\frac{f(n)}{g(n)}=0$; (iii) $f(n)=\Omega(g(n))$ if and only if $\limsup_{n\rightarrow\infty}\frac{f(n)}{g(n)}>0$; (iv) $f(n)=\omega(g(n))$ if and only if $\lim_{n\rightarrow\infty}\frac{f(n)}{g(n)}=\infty$; and (v) $f(n)\sim g(n)$ if and only if $\lim_{n\rightarrow\infty}\frac{f(n)}{g(n)}=1$.

\section{Deterministic Regular Transmissions and Bernoulli Losses}\label{sec:BernoulliLossRegularTraffic}
In this section, for each coding scheme, we first consider arbitrary traffic parameters $\{p_i\}$; and next, we consider a special case with unequal traffic parameters.

\subsection{Dense Codes}\label{subsec:DC}
In a dense coding scheme, the source node, at each transmission opportunity, transmits a packet by randomly linearly combining the message vectors, and each non-source non-sink (interior) node transmits a packet by randomly linearly combining its previously received packets. The vector of coefficients of the linear combination associated with a packet is called the \emph{local encoding vector} of the packet, and the vector of the coefficients representing the mapping between the message vectors and a coded packet is called the \emph{global encoding vector} of the packet. The global encoding vector of each packet is assumed to be included in the packet header. The sink node can recover all the message vectors as long as it receives an innovative collection of packets (with linearly independent global encoding vectors) of the size equal to the number of message vectors at the source node.

The first step in our analysis is to lower bound the size of a maximal collection of packets at any non-source node until a certain time, referred to as the \emph{decoding time}, where the entries of the global encoding vectors of the packets in the collection are independent and uniformly distributed (i.u.d.) Bernoulli random variables. Such packets are called the \emph{globally dense packets}. Based on the result of~\cite[Lemma~1]{HBJ:2011}, the size of a maximal collection of globally dense packets at a node can be lower bounded by the number of packets with linearly independent local encoding vectors at that node. With a slight abuse of terminology, the packets with linearly independent local encoding vectors are called the \emph{dense packets}. (By the above argument, the set of dense packets at each node is a subset of the globally dense packets at that node.\footnote{One should however note that the local encoding vectors being linearly independent (i.e., forming a ``dense'' collection of packets) is not a necessary condition for the packets to form a ``globally dense'' collection. In particular, the collection of all the packets successfully transmitted by the source node is globally dense (by the definition) but some packets in this collection might have local encoding vectors linearly dependent on those of the rest (and hence such packets do not belong to the dense collection of the packets successfully transmitted by the source node).}) The set of dense packets are of main importance in our analysis. In particular, by studying the linear dependence/independence of the local encoding vectors of the successful packets over a link, the number of dense packets, and further, the size of a maximal collection of globally dense packets, at the receiving node of that link can be lower bounded. We, next, upper bound the decoding time such that the probability that the underlying collection fails to include an innovative sub-collection of a sufficiently large size (equal to the message size) is upper bounded (this probability upper bounds the probability of the failure of a dense code to recover all the message packets till the underlying decoding time).

Let $\mathcal{O}_{i}$ ($\mathcal{I}_i$) be the set of labels of the successful (i.e., not lost) packets transmitted (received) by the $i\textsuperscript{th}$ node and let $\mathcal{D}_i$ be the set of labels of the dense packets at the $i\textsuperscript{th}$ node. Let $Q_{i+1}$ and $Q_{i}$, with entries over $\mathbb{F}_2$, be the decoding matrices\footnote{The global encoding vectors of the received packets at a node form the rows of the \emph{decoding matrix} at that node.} at the $(i+1)\textsuperscript{th}$ and $i\textsuperscript{th}$ nodes, respectively, and $T_{i}$, the \emph{transfer matrix} at the $i\textsuperscript{th}$ node, be a matrix over $\mathbb{F}_2$ such that $Q_{i+1}=T_{i}Q_{i}$. The rows of $T_i$ are the local encoding vectors of the successful packets transmitted by the $i\textsuperscript{th}$ node, i.e., $(T_i)_{n,j}=\lambda_{n,j}$, $\forall {n}\in \mathcal{O}_{i}$ and $\forall {j}\in \mathcal{I}_{i}$, where $\lambda_n$ is the local encoding vector of the $n\textsuperscript{th}$ successful packet. Let $\hat{Q}_{i}$, the \emph{modified decoding matrix} at the $i\textsuperscript{th}$ node, be $Q_{i}$ restricted to its rows pertaining to the global encoding vectors of the dense packets. Let $\hat{T}_{i}$, the \emph{modified transfer matrix} at the $i\textsuperscript{th}$ node, be a matrix over $\mathbb{F}_2$ such that $Q_{i+1}=\hat{T}_{i}\hat{Q}_{i}$, i.e., \[(\hat{T}_{i})_{{n},j}={\lambda}_{{n},j}+\sum_{\ell\in \mathcal{I}_{i}\setminus \mathcal{D}_{i}}\lambda_{n,\ell}\gamma_{\ell,j}, \forall n\in \mathcal{O}_i, \forall j\in\mathcal{D}_i\] and $\{\gamma_{\ell,j}\}$ are in $\mathbb{F}_2$ satisfying $\sum_{j\in \mathcal{D}_i}\gamma_{\ell,j}\lambda_{j,k}=\lambda_{\ell,k}$, $\forall k\in\mathcal{I}_{i}$. The $n\textsuperscript{th}$ row of $\hat{T}_{i}$ indicates the labels of dense packets at the $i\textsuperscript{th}$ node which contribute to the $n\textsuperscript{th}$ successful packet sent by the $i\textsuperscript{th}$ node, and the $j\textsuperscript{th}$ column of $\hat{T}_{i}$ indicates the labels of successful packets sent by the $i\textsuperscript{th}$ node to which the $j\textsuperscript{th}$ dense packet contributes.


Let $Q$ be a matrix over $\mathbb{F}_2$. The \emph{density} of $Q$, denoted by $\mathcal{D}(Q)$, is the size of a maximal \emph{dense} collection of rows in $Q$, where a collection of rows is dense if the rows have all i.u.d. entries over $\mathbb{F}_2$. Further, $Q$ is called a \emph{dense matrix} if all its rows form a dense collection. Let $T$ be a matrix over $\mathbb{F}_2$. The \emph{rank} of $T$, denoted by $\text{rank}(T)$, is the size of a maximal collection of linearly independent rows in $T$ over $\mathbb{F}_2$.

\begin{lemma}\label{lem:DensityTM} Let $Q$ be a dense matrix over $\mathbb{F}_2$, and $T$ be a matrix over $\mathbb{F}_2$, where the number of rows in $Q$ and the number of columns in $T$ are equal. If $\text{rank}(T)\geq \gamma$, then $\mathcal{D}(TQ)\geq \gamma$.\end{lemma}

\begin{proof}The proof can be found in~\cite{HBJ:2011}.\end{proof}

Since $Q_{i+1}=\hat{T}_i \hat{Q}_i$, and $\hat{Q}_i$ is dense,\footnote{The rows in $\hat{Q}_i$ are the global encoding vectors of the dense packets at the $i\textsuperscript{th}$ node, and based on an earlier argument, the set of dense packets at a node belong to the set of globally dense packets at that node. Thus, the entries of all the rows in $\hat{Q}_i$ are i.u.d. over $\mathbb{F}_2$.} by applying the result of Lemma~\ref{lem:DensityTM}, it follows that $\mathcal{D}(Q_{i+1})$ can be bounded from below so long as $\text{rank}(\hat{T}_i)$ is bounded from below. The rank of the modified transfer matrix $\hat{T}_i$ is a function of the structure of $\hat{T}_i$, and the structure of such a matrix depends on the number of dense packet arrivals at the $i\textsuperscript{th}$ node and the number of successful packet departures from the $i\textsuperscript{th}$ node before or after any given point in time. Such parameters depend on the traffic over the $i\textsuperscript{th}$ and $(i+1)\textsuperscript{th}$ links, and are therefore random variables. It is however not straightforward to find the distribution of such random variables. We thus use a probabilistic technique as follows to study such variables.

Let $(0,N_T]$ be the period of time over which the transmissions occur ($N_T$ is the decoding time). We split the time interval $(0,N_T]$ into $w$ disjoint subintervals (partitions) of length $N_T/w$. The first partition represents the time interval $(0,N_T/w]$; the second partition represents the time interval $(N_T/w,2 N_T/w]$, and so forth. For every $1\leq i< L$ and $1\leq j<w$, all the arrivals over the $i\textsuperscript{th}$ link in the first $j$ partitions, i.e., in the time interval $(0,j N_T/w]$, occur before any departure over the $(i+1)\textsuperscript{th}$ link in the $(j+1)\textsuperscript{th}$ partition, i.e., in the time interval $(j N_T/w,(j+1)N_T/w]$. Thus the number of arrivals at the $i\textsuperscript{th}$ node before any given point in time within the $(j+1)\textsuperscript{th}$ partition is bounded from below by the sum of the number of arrivals at this node in the first $j$ partitions.

This method of counting is however suboptimal since there might be some extra arrivals in the $(j+1)\textsuperscript{th}$ partition, which arrive before the given point in time within this partition. To control the impact of sub-optimality, the length of partitions needs to be chosen with some care. To be specific, the length of partitions, on the one hand, needs to be sufficiently small such that there is not a large number of arrivals in one partition compared to the total number of arrivals in all the partitions. This should be the case so that ignoring a subset of arrivals in one partition does not result in a significant difference in the number of arrivals before each point in time within the same partition. On the other hand, the partitions need to be long enough such that the deviation of the number of arrivals from the expectation in one partition is negligible in comparison with the expectation itself. This ensures the validity of our analysis and the tightness of our results.

Let $I_{ij}$ represent the $j\textsuperscript{th}$ partition pertaining to the $i\textsuperscript{th}$ link for every $i$ and $j$. We focus on the set of all the packets over the $i\textsuperscript{th}$ link in the \emph{active partitions} pertaining to this link, where, for every $i,j$, $I_{ij}$ is an active partition if and only if $i\leq j\leq w-L+i$. Such a partition is active in the sense that (i) there exists some other partition over the upper link so that all its packets arrive before the departure of any packet in the underlying active partition, and (ii) there exists some other partition over the lower link so that all its packets depart after the arrival of any packet in the underlying active partition. In particular, the first $w-L+1$ partitions pertaining to the first link are all active; the $w-L+1$ partitions pertaining to the second link starting from the second partition are all active and so forth. Let $w_T$ represent the total number of active partitions pertaining to all the links, i.e., \begin{equation}\label{eq:wT}w_T\doteq L(w-L+1).\end{equation}

We start off with lower bounding the number of successful packets in all the $w_T$ active partitions. Let $I_{ij}$ be an active partition, and $\varphi_{ij}$ be the number of (successful) packets in $I_{ij}$. Since the length of the partition $I_{ij}$ is $N_{T}/w$, and by the assumption the packet successes over the $i\textsuperscript{th}$ link follow a Bernoulli process with the parameter $p_i$, $\varphi_{ij}$ is a binomial random variable with the expected value $\varphi_i \doteq p_i N_T /w$. Let \begin{equation}\label{eq:p}p\doteq\min_{1\leq i\leq L}p_i,\end{equation} and \begin{equation}\label{eq:varphi}\varphi\doteq pN_T/w.\end{equation} For any real number $x$, let $\dot{x}$ denote $\frac{x}{2}$. By applying the Chernoff bound, one can show that $\varphi_{ij}$ is not larger than or equal to \begin{equation}\label{eq:r}r\doteq (1-\gamma^{*})\varphi\end{equation} w.p. b.a.b. $\dot{\epsilon}/w_T$, so long as $0<\gamma^{*}<1$ is chosen such that $r$ is an integer, and $\gamma^{*}$ goes to $0$ as $N_T$ goes to infinity, where \begin{equation}\label{eq:gammastar}\gamma^{*}\sim \left(\frac{1}{\dot{\varphi}}\ln\frac{w_T}{\dot{\epsilon}}\right)^{\frac{1}{2}}.\end{equation}

For all $i,j$, suppose that $\varphi_{ij}$ is larger than or equal to $r$. This assumption fails if the number of packets in some active partition is less than $r$. Hence, the failure occurs w.p. b.a.b. $\dot{\epsilon}$.

Next, for every $1<i\leq L$ and $1\leq j\leq w-L+1$, we lower bound the number of dense packets in the first $j$ active partitions over the $i\textsuperscript{th}$ link. Before explaining the lower bounding technique in detail, let us introduce two lemmas which will be useful to lower bound the rank of the modified transfer matrix at each node (depending on whether the number of dense packet arrivals at the $i\textsuperscript{th}$ node over the $i\textsuperscript{th}$ link in a given partition is larger or smaller than the number of packet departures from that node over the $(i+1)\textsuperscript{th}$ link in the partition with the same index as the underlying partition pertaining to the $i\textsuperscript{th}$ link).

Let $w^{*}$, $r^{*}$ and $\{r^{*}_l\}_{1\leq l\leq w^{*}}$ be arbitrary non-negative integers, and let $r^{*}_{\text{max}}=\max_{1\leq l\leq w^{*}} r^{*}_l$ and $r^{*}_{\text{min}}=\min_{1\leq l\leq w^{*}} r^{*}_l$. For any pair $(i',j')$ such that $1\leq j'\leq i'\leq w^{*}$, let $T_{i',j'}$ be an $r^{*}\times r^{*}_{j'}$ dense matrix over $\mathbb{F}_2$; for any other pair $(i',j')$, let $T_{i',j'}$ be an arbitrary $r^{*}\times r^{*}_{j'}$ matrix over $\mathbb{F}_2$.\footnote{For any pair $(i',j')$ such that $1\leq i' < j' \leq w^{*}$, the entries of $T_{i',j'}$ might be dependent on the entries of $T_{i'',j''}$, for any other pair $(i'',j'')$ such that $1\leq j''\leq i''\leq w^{*}$.} Let $T=[T_{i',j'}]_{1\leq i',j'\leq w^{*}}$. The matrix $T$ is called \emph{random block lower-triangular} (RBLT) (see Figure~\ref{fig:MTM}).

\begin{figure}[t!]
  \centering
  \includegraphics[width=2.5in]{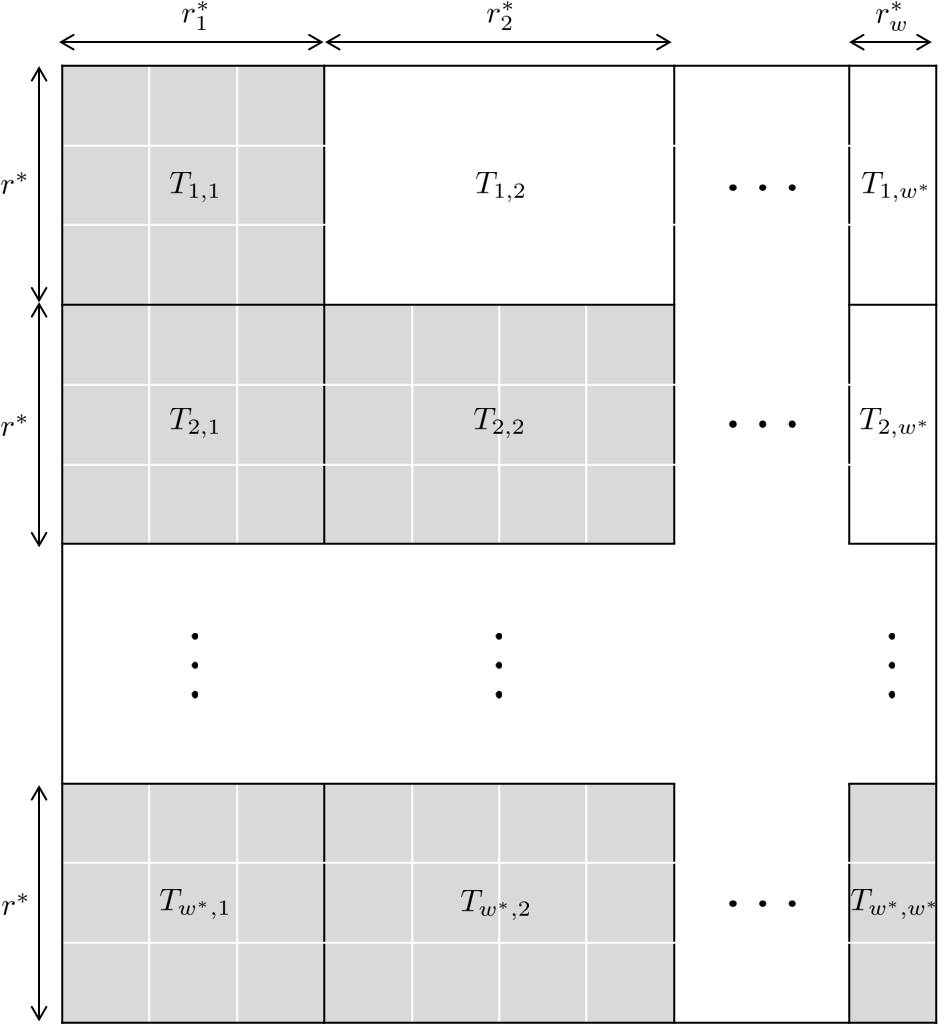}
  \caption{The structure of a random block lower-triangular (RBLT) matrix $T$ with parameters $w^{*},r^{*}$ and $\{r^{*}_l\}_{1\leq l\leq w^{*}}$. The shaded blocks represent the dense sub-matrices of $T$ with i.u.d. Bernoulli entries, and the blank blocks represent those sub-matrices of $T$ with arbitrarily dependent or independent entries with respect to the entries of the dense sub-matrices of $T$.}
  \label{fig:MTM}
\end{figure}

\begin{lemma}\label{lem:VerticalT} Let $T$ be an RBLT matrix with parameters $w^{*}$, $r^{*}$ and $\{r^{*}_{l}: 0\leq r^{*}_{l}\leq r^{*}\}_{1\leq l\leq w^{*}}$. Let $n^{*}= \sum_{1\leq l\leq w^{*}}r^{*}_{l}$. For every integer $0\leq \gamma\leq n^{*}-1$, \[\Pr\{\text{rank}(T)<n^{*}-\gamma\}\leq u^{*} \left(1-2^{-r^{*}_{\text{max}}}\right) 2^{-\gamma+n^{*}-w^{*}r^{*}+(r^{*}-r^{*}_{\text{min}})(u^{*}-1)},\] where $u^{*}=\left\lceil{(n^{*}-\gamma)}/{r^{*}_{\text{min}}}\right\rceil$.\end{lemma}

\begin{proof}For any integer $0\leq \gamma\leq n^{*}-1$, let $T'$ be $T$ restricted to its first $n^{*}-\gamma$ columns. Since $T'$ is an $w^{*}r^{*}\times (n^{*}-\gamma)$ sub-matrix of $T$, $\Pr\{\text{rank}(T)<n^{*}-\gamma\}\leq \Pr\{\text{rank}(T')<n^{*}-\gamma\}$. Suppose that $\text{rank}(T')<n^{*}-\gamma$. Then there exists a nonzero column vector $\boldsymbol{v}$ of length $n^{*}-\gamma$ over $\mathbb{F}_2$ such that the column vector $T'\boldsymbol{v}$ of length $w^{*}r^{*}$ is an all-zero vector. For a given integer $1\leq j\leq n^{*}-\gamma$, suppose that the first non-zero entry of $\boldsymbol{v}$ is the $j\textsuperscript{th}$. There exists $2^{n^{*}-\gamma-j}$ such vectors. Let us define $r^{*}_0\doteq 0$ for convenience. Let $\tau$ be an integer satisfying $\sum_{0\leq i\leq \tau}r^{*}_i<j\leq \sum_{0\leq i\leq \tau+1}r^{*}_i$, and $\tau_{\text{max}}$ be an integer satisfying $\sum_{0\leq i\leq \tau_{\text{max}}}r^{*}_i<n^{*}-\gamma\leq\sum_{0\leq i\leq \tau_{\text{max}}+1}r^{*}_i$. By the definition, it follows that $0\leq \tau_{\text{max}}\leq \min\{w^{*},u^{*}-1\}$. It should not be hard to see that $\tau$ and $\tau_{\text{max}}$ are unique. For every $0\leq \tau\leq \tau_{\text{max}}$, define $s^{*}_{\tau}=\sum_{0\leq i\leq \tau}r^{*}_i$. The $j\textsuperscript{th}$ column of $T'$ has at least $(w^{*}-\tau)r^{*}$ i.u.d. Bernoulli entries, and hence the vector $T'\boldsymbol{v}$ has at least $(w^{*}-\tau)r^{*}$ i.u.d. Bernoulli entries. Thus, $T'\boldsymbol{v}$ is all-zero w.p. b.a.b. $2^{-\gamma+n^{*}-w^{*}r^{*}}\sum_{1\leq j\leq n^{*}-\gamma}2^{\tau r^{*}-j}$, noting that $\tau$ depends on $j$. We rewrite the sum as:
\begin{eqnarray*}
&& \hspace*{-1 cm} \sum_{0<j\leq s^{*}_1} 2^{-j}+\sum_{s^{*}_1<j\leq s^{*}_2}2^{r^{*}-j}+ \\
& & \cdots+\sum_{s^{*}_{\tau_{\text{max}}}<j\leq n^{*}-\gamma}2^{\tau_{\text{max}}r^{*}-j} = \\
& & \sum_{0<j\leq r^{*}_1}2^{-j}+2^{r^{*}-s^{*}_1}\sum_{0<j\leq r^{*}_2}2^{-j}+ \\
& & \cdots+2^{\tau_{\text{max}}r^{*}-s^{*}_{\tau_{\text{max}}}}\sum_{0<j\leq n^{*}-\gamma-s^{*}_{\tau_{\text{max}}}}2^{-j}\leq \\
& & \sum_{0<j\leq r^{*}_{\text{max}}}2^{-j} + 2^{r^{*}-s^{*}_1}\sum_{0<j\leq r^{*}_{\text{max}}}2^{-j}+ \\
& & \cdots+2^{\tau_{\text{max}}r^{*}-s^{*}_{\tau_{\text{max}}}}\sum_{0<j\leq r^{*}_{\text{max}}}2^{-j}= \\
& & \sum_{0<j\leq r^{*}_{\text{max}}}2^{-j}\sum_{0\leq \tau'\leq \tau_{\text{max}}}2^{\tau' r^{*}-s^{*}_{\tau'}}\leq \\
& & \sum_{0<j\leq r^{*}_{\text{max}}}2^{-j}\sum_{0\leq \tau'\leq \tau_{\text{max}}}2^{(r^{*}-r^{*}_{\text{min}})\tau'}= \\
& & (1-2^{-r^{*}_{\text{max}}})\sum_{0\leq \tau'\leq \tau_{\text{max}}}2^{(r^{*}-r^{*}_{\text{min}})\tau'}.
\end{eqnarray*} The series $\sum_{0\leq \tau'\leq \tau_{\text{max}}}2^{(r^{*}-r^{*}_{\text{min}})\tau'}$ converges from below to $(\tau_{\text{max}}+1)2^{(r^{*}-r^{*}_{\text{min}})\tau_{\text{max}}}$ if $r^{*}-r^{*}_{\text{min}}$ goes to infinity. Thus the following is always true: $(1-2^{-r^{*}_{\text{max}}})$ $\sum_{0\leq \tau'\leq \tau_{\text{max}}}2^{(r^{*}-r^{*}_{\text{min}})\tau'}$ $\leq$ $(\tau_{\text{max}}+1)(1-2^{-r^{*}_{\text{max}}})2^{(r^{*}-r^{*}_{\text{min}})\tau_{\text{max}}}$ $\leq u^{*}(1-2^{-r^{*}_{\text{max}}})2^{(r^{*}-r^{*}_{\text{min}})(u^{*}-1)}$. This proves the lemma.\end{proof}

\begin{lemma}\label{lem:HorizontalT} Let $T$ be an RBLT matrix with parameters $w^{*}$, $r^{*}$ and $\{r^{*}_{l}: 0\leq r^{*}\leq r^{*}_{l}\}_{1\leq {l}\leq w^{*}}$. Let $n^{*} = w^{*}r^{*}$. For every integer $0\leq \gamma\leq n^{*}-1$, \[\Pr\{\text{rank}(T)<n^{*}-\gamma\}\leq u^{*} \left(1-2^{-r^{*}}\right) 2^{-\gamma+n^{*}-w^{*}r^{*}_{\text{min}}+(r^{*}_{\text{min}}-r^{*})(u^{*}-1)},\] where $u^{*}=\left\lceil{(n^{*}-\gamma)}/{r^{*}}\right\rceil$.\end{lemma}

\begin{proof}We start the proof by noting that $T$ has a smaller number of rows than columns, and the minimum number of rows and columns gives an upper bound on the rank of the matrix. Let $T'$ be $T$ restricted to its last $n^{*}-\gamma$ rows. For every $0\leq\tau\leq w^{*}$, define $s^{*}_{\tau}=\sum_{0\leq j\leq w^{*}-\tau}r^{*}_j$. Thus, $T'$ is of size ${(n^{*}-\gamma)\times s^{*}_0}$. Suppose that there exists a nonzero row vector $\boldsymbol{v}$ of length $n^{*}-\gamma$ whose entries are over $\mathbb{F}_2$, and its first nonzero entry is the $j\textsuperscript{th}$, and the row vector $\boldsymbol{v}T'$ is all-zero. There are $2^{n^{*}-\gamma-j}$ such vectors. Let $\tau$ be the largest integer smaller than $j/r^{*}$. The $j\textsuperscript{th}$ row of $T'$ has at least $s^{*}_{\tau}$ i.u.d. Bernoulli entries, and hence the vector $\boldsymbol{v}T'$ has at least $s^{*}_{\tau}$ i.u.d. Bernoulli entries. Thus, $\boldsymbol{v}T'$ is all-zero w.p. b.a.b. $2^{-\gamma+n^{*}}\sum_{1\leq j\leq n^{*}-\gamma}2^{-j-s^{*}_{\tau}}$. By definition, $s^{*}_{\tau}\geq (w^{*}-\tau)r^{*}_{\text{min}}$, and the preceding sum can thus be upper bounded as follows: $\sum_{1\leq j\leq n^{*}-\gamma}2^{-j-s^{*}_{\tau}}\leq$ $\sum_{1\leq j\leq n^{*}-\gamma}2^{-j-(w^{*}-\tau)r^{*}_{\text{min}}}$. The latter sum can be rewritten itself as:
\begin{eqnarray*}&& \hspace*{-1 cm} \sum_{0<j\leq r^{*}}2^{-j-w^{*}r^{*}_{\text{min}}}+\sum_{r^{*}<j\leq 2r^{*}}2^{-j-(w^{*}-1)r^{*}_{\text{min}}}+\\
& & \cdots +\sum_{(u^{*}-1)r^{*}<j\leq n^{*}-\gamma}2^{-j-(w^{*}-u^{*}+1)r^{*}_{\text{min}}} =  \\
& & 2^{-w^{*}r^{*}_{\text{min}}}\sum_{0<j\leq r^{*}}2^{-j}+2^{-(w^{*}-1)r^{*}_{\text{min}}-r^{*}}\sum_{0<j\leq r^{*}}2^{-j}+ \\
& & \cdots+2^{-(w^{*}-1)r^{*}_{\text{min}}-(u^{*}-1)r^{*}}\sum_{0<j\leq n^{*}-\gamma-(u^{*}-1)r^{*}}2^{-j} \leq \\
& & 2^{-w^{*}r^{*}_{\text{min}}}\sum_{0<j\leq r^{*}}2^{-j}\sum_{0\leq \tau'\leq u^{*}-1}2^{(r^{*}_{\text{min}}-r^{*})\tau'}= \\
& & (1-2^{-r^{*}})2^{-w^{*}r^{*}_{\text{min}}}\sum_{0\leq \tau'\leq u^{*}-1}2^{(r^{*}_{\text{min}}-r^{*})\tau'}.\end{eqnarray*} The last sum is bounded from above by $2^{(r^{*}_{\text{min}}-r^{*})(u^{*}-1)}u^{*}$, and this completes the proof.\end{proof}

Now, for every $1<i\leq L$ and $1\leq j\leq w-L+1$, we explain how to lower bound the number of dense packets in the first $j$ active partitions over the $i\textsuperscript{th}$ link. The lower bounding technique works in a recursive manner as follows:

For every $1\leq l\leq j$, suppose that the number of dense packets in the first $l$ active partitions over the $(i-1)\textsuperscript{th}$ link is lower bounded. Let $\hat{T}_{i}^{j}$ be the modified transfer matrix at the $i\textsuperscript{th}$ node, restricted to the successful packet transmissions within the first $j$ active partitions over the $i\textsuperscript{th}$ link (by the assumption, the number of such packets in each partition is bounded from below by $r$). Then, one can see that the matrix $\hat{T}_{i}^{j}$ includes a sub-matrix $\hat{T}'$ with a structure similar to that in Lemma~\ref{lem:VerticalT} or the one in Lemma~\ref{lem:HorizontalT}.\footnote{In the case of identical links, the modified transfer matrix at each node includes a sub-matrix similar to that in Lemma~\ref{lem:VerticalT}. However, in the case of non-identical links, depending on the traffic parameters, the modified transfer matrix at a node might include a sub-matrix similar to that in Lemma~\ref{lem:VerticalT} or the one in Lemma~\ref{lem:HorizontalT}.} This can be seen precisely by the following replacements in Lemma~\ref{lem:VerticalT} or Lemma~\ref{lem:HorizontalT}: (i) $w^{*}$ with $j$ (i.e., the number of underlying active partitions), (ii) $r^{*}$ with $r$ (i.e., the lower bound on the number of successful packet transmissions in each of the first $j$ active partitions pertaining to the $i\textsuperscript{th}$ link), and (iii) $r^{*}_l$, for every ${1\leq l\leq w^{*}}$, with the difference between the two lower bounds on the number of dense packets in the first $l$ and the first $l-1$ active partitions pertaining to the $(i-1)\textsuperscript{th}$ link (note that $r_1$ is equal to the lower bound on the number of dense packets in the first active partition).

The lower bounding process then proceeds as follows. Each successful packet in any of the first $j$ active partitions, say the $m\textsuperscript{th}$ active partition, for some $1\leq m\leq j$, pertaining to the $i\textsuperscript{th}$ link can be written as a linear combination of the dense packets in the $l\textsuperscript{th}$ active partition pertaining to the $(i-1)\textsuperscript{th}$ link, for all $1\leq l\leq m$, and perhaps some extra dense packets in the $(m+1)\textsuperscript{th}$ active partition pertaining to the $(i-1)\textsuperscript{th}$ link. Thus, for every $1\leq l\leq m$, each row of the sub-matrix $\hat{T}'_{i-1,l}$ (in the matrix $\hat{T}'$) indicates the labels of (some subset of)\footnote{It is worth noting that there might be a number of dense packets which contribute to the linear combination of some packet transmission, but are not included in our lower bounding analysis. The exclusion of such (dense) packets weakens the tightness of the results, but does not affect the correctness of the analysis.} the dense packets in the $l\textsuperscript{th}$ active partition pertaining to the $(i-1)\textsuperscript{th}$ link which contribute to the linear combination of one packet (from the set of the $r$ chosen successful packets) in the $m\textsuperscript{th}$ active partition pertaining to the $i\textsuperscript{th}$ link; and each column of $\hat{T}'_{i-1,l}$ indicates the labels of the successful packets to which one dense packet (from the set of the $r_l$ chosen dense packets) in the $l\textsuperscript{th}$ active partition pertaining to the $(i-1)\textsuperscript{th}$ link contributes. For any other $m<l\leq w$, the set of dense packets in the $l\textsuperscript{th}$ active partition pertaining to the $(i-1)\textsuperscript{th}$ link which contribute to the linear combination of one packet in the $m\textsuperscript{th}$ active partition pertaining to the $i\textsuperscript{th}$ link is not tractable in our analysis. Hence, the rows (or the columns) of each sub-matrix $\hat{T}'_{i-1,l}$, for such values of $l$ (i.e., $m<l\leq w$), might have independent or dependent entries (over $\mathbb{F}_2$) with respect to the entries of the rows (or the columns) of the sub-matrices in the set of $\{\hat{T}'_{i-1,l}\}_{1\leq l\leq m}$. Next, by applying the proper lemma, the rank of the modified transfer matrix at the $i\textsuperscript{th}$ node, and finally, by applying Lemma~\ref{lem:DensityTM}, the number of dense packets in the first $j$ active partitions over the $i\textsuperscript{th}$ link can be bounded from below. This completes the lower bounding process.

Note that, because of its recursive nature, the above technique lower bounds the number of dense packets in the first $j$ active partitions over the $i\textsuperscript{th}$ link as a function of the number of dense packets in the active partitions pertaining to the first link. Further, the packets over the first link are all globally dense (by the definition), and hence by using the recursion, the required results can be derived as follows.

Let $\mathcal{D}(Q_i^j)$ be the number of ``globally dense'' packets in the first $j$ active partitions over the $i\textsuperscript{th}$ link. By the definition, $\mathcal{D}(Q_i^j)$ is bounded from below by the number of ``dense'' packets in the first $j$ active partitions over the $i\textsuperscript{th}$ link. Let $\mathcal{D}_{p}(Q_i^j)$ be a (probabilistic) lower bound on the number of dense packets in the first $j$ active partitions over the $i\textsuperscript{th}$ link, and of course a lower bound on $\mathcal{D}(Q_i^j)$,\footnote{$\mathcal{D}_{p}(Q_i^j)$ is a ``probabilistic'' lower bound on $\mathcal{D}(Q_i^j)$ and hence the subscript ``$p$.''} such that: if $\mathcal{D}(Q_{s}^{\tau})\geq\mathcal{D}_{p}(Q_{s}^{\tau})$, for every $1\leq s\leq i$ and $1\leq \tau\leq j$, except $(s,\tau)=(i,j)$, then the inequality $\mathcal{D}(Q_i^j)\geq\mathcal{D}_{p}(Q_i^j)$ fails w.p. b.a.b. $\dot{\epsilon}/w_T$. Let $\hat{r}_{ij}$ be defined in a recursive fashion as the largest integer satisfying \begin{equation}\label{eq:rij}\hat{r}_{ij}\leq \mathcal{D}_{p}(Q_i^j)-\sum_{1\leq \tau<j}\hat{r}_{i\tau}.\end{equation}

Note that, at each step of our lower bounding process, the number of dense packets in a collection of active partitions, but not the number of dense packets in one individual active partition, is lower bounded. Further, the difference between the two lower bounds corresponding to the two collections of the first $j$ and the first $j-1$ active partitions does not lower bound the number of dense packets in the $j\textsuperscript{th}$ active partition. However, due to the recursion, we need to choose a certain number of dense packets at each step of the process (and ignore the rest, if any), and study the density of the packets in the next partition, at the next step of the process, with respect to the dense packets chosen till the previous step. We, thus, construct a collection of dense packets at the $i\textsuperscript{th}$ node as follows: starting with an empty collection (at the step zero), for every $1\leq j\leq w-L+1$, at the $j\textsuperscript{th}$ step, we expose the packets in the active partitions over the $i\textsuperscript{th}$ link in order, one by one. We add a packet to the collection whenever the packet is dense (with respect to the current collection), until revealing $\hat{r}_{ij}$ new dense packets. The size of such a collection lower bounds the number of dense packets at the $i\textsuperscript{th}$ node; and in order to study the structure of the modified transfer matrix at this node, we fix the packets in the subsets of the underlying collection (each subset pertaining to one of the collection steps) and ignore the rest of packets.

The set of packets over the first link are globally dense, and hence, for every $1\leq j\leq w-L+1$, $\mathcal{D}(Q_1^j)\geq rj$ (by the assumption, each partition includes more than or equal to $r$ packets). Further, for every $1<i\leq L$ and $1\leq j\leq w-L+1$, $\mathcal{D}(Q_i^j)$ is bounded from below as follows.

\begin{lemma}\label{lem:Omegaii}Consider applying a dense code over a line network of $L$ links with deterministic regular transmissions and Bernoulli losses with parameters $\{p_i\}$. Let $w_T$ and $r$ be defined as in~\eqref{eq:wT} and~\eqref{eq:r}, respectively. For every $1< i\leq L$, the inequality \[\mathcal{D}(Q_i^1)\geq r-\log(w_T/\epsilon)- 1\] fails w.p. b.a.b. $i\dot{\epsilon}/w_T$.\end{lemma}

\begin{proof}Fix $1<i\leq L$. Let $\hat{T}_{i-1}$ be the modified transfer matrix at the starting node of the $i\textsuperscript{th}$ link. Let $\hat{T}^{1}_{i-1}$ be $\hat{T}_{i-1}$ restricted to the packets in the first active partition over the $i\textsuperscript{th}$ link. For every $1<s<i$, suppose $\mathcal{D}(Q_s^1)\geq \mathcal{D}_p(Q_s^1)$, where $\mathcal{D}_p(Q_s^1)=r-\log(w_T/\epsilon)-1$, and $\mathcal{D}(Q_1^1)=\mathcal{D}_p(Q_1^1)=r$. Then, by replacing $w^{*},r^{*}$ and $\{r^{*}_l\}_{1\leq l\leq w^{*}}$ with $1,r$ and $\hat{r}_1$, respectively, in Lemma~\ref{lem:VerticalT}, where $\hat{r}_1\doteq \hat{r}_{i-1,1}$,\footnote{We often drop the subscript $i$ in the notation $r_{ij}$ when there is no danger of confusion.} one can see that $\hat{T}^{1}_{i-1}$ includes an $r\times \hat{r}_1$ dense sub-matrix. Thus by applying Lemma~\ref{lem:VerticalT}, for every $0\leq \gamma\leq \hat{r}_1-1$, $\Pr\{\text{rank}(\hat{T}^{1}_{i-1})<\hat{r}_1-\gamma\}\leq u (1-2^{-\hat{r}_1}) 2^{-\gamma+\hat{r}_1-r+(r-\hat{r}_1)(u-1)}$, where $u=\lceil(\hat{r}_1-\gamma)/\hat{r}_1\rceil$. Thus, $\Pr\{\text{rank}(\hat{T}^{1}_{i-1})<(1-2^{-\hat{r}_1})2^{-\gamma+\hat{r}_1-r}$, since $u=1$. Taking $\gamma=\log(w_T/\epsilon)+\hat{r}_1-r+1$, it follows that $\Pr\{\text{rank}(\hat{T}^{1}_{i-1})<\hat{r}_1-\gamma\}\leq \dot{\epsilon}/w_T$. By Lemma~\ref{lem:DensityTM}, $\mathcal{D}(Q_i^1)<r-\log(w_T/\epsilon)-1$ w.p. b.a.b. $\dot{\epsilon}/w_T$. Thus, $\mathcal{D}_p(Q_i^1)=r-\log(w_T/\epsilon)-1$. Taking a union bound over the first $i$ links, $\mathcal{D}(Q_i^1)<r-\log(w_T/\epsilon)-1$ w.p. b.a.b. $i\dot{\epsilon}/w_T$.\end{proof}

\begin{lemma}\label{lem:Omegaij}Consider a scenario similar to the one in Lemma~\ref{lem:Omegaii}. Let $p$ be defined as in~\eqref{eq:p}. For every $1<i\leq L$ and $1<j\leq w-L+1$, the inequality \[\mathcal{D}(Q_i^j)\geq rj-\mathcal{L}_{ij}\] fails w.p. b.a.b. $ij\dot{\epsilon}/w_T$, so long as \begin{equation}\label{eq:Temp1}w\log\frac{w_T}{\epsilon}=o(p N_T)\end{equation} where $\mathcal{L}_{ij} = j(1+o(1))(\log(w_T/\epsilon)+1)+\log((j(1+o(1))+1)/\epsilon)+\log w_T+1$, and the $o(1)$ term is $(\log(w_T/\epsilon)+1)/r$.\end{lemma}

\begin{proof}Fix $1<i\leq L$. For every $1<s\leq i$ and $1<\tau\leq j$, except $(s,\tau)=(i,j)$, suppose $\mathcal{D}(Q_s^{\tau})\geq \mathcal{D}_p(Q_s^{\tau})$, where $\mathcal{D}_p(Q_s^{\tau})=r\tau-\tau(1+o(1))(\log(w_T/\epsilon)+1)-\log((\tau(1+o(1))+1)/\epsilon)-\log w_T-1$, and the $o(1)$ term is $(\log(w_T/\epsilon)+1)/r$, and $\mathcal{D}(Q_s^{1})\geq \mathcal{D}_p(Q_s^{1})$, where $\mathcal{D}_p(Q_s^{1})=r-\log(w_T/\epsilon)-1$. Let $\hat{r}_{\tau}=\hat{r}_{i-1,\tau}$, for every $1\leq \tau\leq j$, $\hat{r}_{\text{min}}=\min_{\tau}\hat{r}_{\tau}$, and $\hat{r}_{\text{max}}=\max_{\tau}\hat{r}_{\tau}$. Let and $n=\mathcal{D}_p(Q_{i-1}^{j})=\sum_{1\leq\tau\leq j}\hat{r}_{\tau}$. Let us define $\hat{T}_{i-1}$ as in the proof of Lemma~\ref{lem:Omegaii}. Let $\hat{T}^{j}_{i-1}$ be $\hat{T}_{i-1}$ restricted to the packets in the first $j$ active partitions over the $i\textsuperscript{th}$ link. Then, by replacing $w^{*},r^{*}$ and $\{r^{*}_l\}_{1\leq l\leq w^{*}}$ with $j,r$ and $\{\hat{r}_{\tau}\}_{1\leq\tau\leq j}$, respectively, in Lemma~\ref{lem:VerticalT}, one can see that $\hat{T}^{j}_{i-1}$ includes an $rj\times n$ sub-matrix with a structure similar to the matrix $T$ as in Lemma~\ref{lem:VerticalT}. Thus by applying Lemma~\ref{lem:VerticalT}, for every $0\leq \gamma\leq n-1$, $\Pr\{\text{rank}(\hat{T}^{j}_{i-1})<n-\gamma\}\leq u(1-2^{-\hat{r}_{\text{max}}})2^{-\gamma+n-rj+(r-\hat{r}_{\text{min}})(u-1)}$, where $u=\lceil (n-\gamma)/\hat{r}_{\text{min}}\rceil$. It is not difficult to see that, by our method of collecting the dense packets, it follows that $\hat{r}_{\text{min}}=\hat{r}_1$. Further by applying Lemma~\ref{lem:Omegaii}, $\hat{r}_1=\mathcal{D}_p(Q_{i-1}^1)=r-\log(w_T/\epsilon)-1$. Thus, $u\leq \lceil rj/\hat{r}_{1}\rceil=\lceil (1+o(1))j\rceil\leq (1+o(1))j+1$, since $\hat{r}_1=r(1-o(1))$, given $\log(w_T/\epsilon)=o(r)$, where the $o(1)$ term is $(\log(w_T/\epsilon)+1)/r$. Since $r\sim \varphi=p N_T/w_T$, the latter condition can be written as $w\log(w_T/\epsilon)=o(p N_T)$. Taking $\gamma=n-rj+(1+o(1))j(\log(w_T/\epsilon)+1)+\log(((1+o(1))j+1)/\epsilon)+\log w_T+1$, it follows that $\Pr\{\text{rank}(\hat{T}^{j}_{i-1})<n-\gamma\}\leq \dot{\epsilon}/w_T$. Now, by applying Lemma~\ref{lem:DensityTM}, $\mathcal{D}(Q_i^j)<n-\gamma$ w.p. b.a.b. $\dot{\epsilon}/w_T$. Thus, $\mathcal{D}_p(Q_i^j)=n-\gamma$. Taking a union bound over the first $j$ active partitions of the first $i$ links, $\mathcal{D}(Q_i^j)<rj-(1+o(1))j(\log(w_T/\epsilon)+1)-\log(((1+o(1))j+1)/\epsilon)-\log w_T-1$ w.p. b.a.b. $ij\dot{\epsilon}/w_T$, where the $o(1)$ term is $(\log(w_T/\epsilon)+1)/r$. This completes the proof.\end{proof}

The result of Lemma~\ref{lem:Omegaij} lower bounds the number of dense packets at the sink node, $\mathcal{D}(Q_L)$, as follows.

\begin{lemma}\label{lem:DensityBound}Consider a scenario similar to the one in Lemma~\ref{lem:Omegaii}. Let $p$ and $\varphi$ be defined as in~\eqref{eq:p} and~\eqref{eq:varphi}, respectively. The inequality \begin{eqnarray}\label{eq:DensityBound}
 \lefteqn{\mathcal{D}(Q_L) \geq w_T\varphi/L -w_T\varphi/L \sqrt{(1/\dot{\varphi})\log(w_T/\dot{\epsilon})} } \nonumber\\
   && \hspace{0.75cm}{ }-(w_T/L)\log(w_T/\dot{\epsilon})-(w_T/L\varphi)\log^{2}(w_T/\epsilon) \nonumber \\
   && \hspace{0.75cm}{ }-(w_T/L\varphi)\log(w_T/\epsilon) -\log(w_T/{\epsilon}) \nonumber\\
   && \hspace{0.75cm}{ }-\log(w_T/L) - 1
\end{eqnarray} fails w.p. b.a.b. $\epsilon$, so long as \[w\log\frac{w_T}{\epsilon}=o(p N_T)\] where $w\sim\left(p N_T L^2/\log(p N_T L/\epsilon)\right)^{\frac{1}{3}}$.\end{lemma}

\begin{proof}For the ease of exposition, let $v=w_T/L$. Lemma~\ref{lem:Omegaij} gives a lower bound on $\mathcal{D}(Q_L^{v})$. Thus, we can write: $\mathcal{D}(Q_L)\geq\mathcal{D}(Q_L^{v})\geq rv-v(1+o(1))\left(\log(w_T/\epsilon)+1\right)-\log((v (1+o(1)))/\epsilon)-\log w_T-1$, where the $o(1)$ term is $\left(\log(w_T/\epsilon)\right)/r$. This bound fails w.p. b.a.b. $\dot{\epsilon}$, given the assumption that the number of packets in each active partition is larger than or equal to $r$. Since this assumption fails w.p. b.a.b. $\dot{\epsilon}$, the lower bound on $\mathcal{D}(Q_L)$ fails w.p. b.a.b. $\epsilon$. Further, $r=(1-o(1))\varphi$, where the $o(1)$ term is $\sqrt{(1/\dot{\varphi})\ln(w_T/\dot{\epsilon})}$. Thus, $\mathcal{D}(Q_L)\geq \varphi v -o(\varphi v)-v\log(w_T/\epsilon)-v-o(v\log(w_T/\epsilon))-o(v)-\log(v/\epsilon)-\log w_T-1$ fails w.p. b.a.b. ${\epsilon}$, where $o(\varphi v)\sim O(\varphi v \sqrt{(1/{\varphi})\log(w_T/{\epsilon})})$, and $o(v)\sim (v/\varphi)\log(w_T/\epsilon)$. By considering the dominant terms, the right-hand side of the last inequality can be written as \begin{equation}\label{eq:Temp3}\varphi v- O(v \sqrt{\varphi\log(w_T/\epsilon)})-O(v\log(w_T/\epsilon)).\end{equation} We now replace $\varphi$ and $v$ by $p N_T/w$ and $w$ ($v\sim w$), respectively, and rewrite~\eqref{eq:Temp3} as \begin{eqnarray}\label{eq:Temp2} && \hspace*{-1 cm} p N_T - O(p N_T L/w) - \nonumber \\
& & O(\sqrt{p N_T w\log(w L/\epsilon)})-O(w\log(w L/\epsilon)),\end{eqnarray} by using the fact that $w_T$ is $O(w L)$. We select $w$ to be \[\sqrt[3]{\frac{p N_T L^2}{\log(p N_T L/\epsilon)}}\] in order to maximize~\eqref{eq:Temp2} subject to condition~\eqref{eq:Temp1}. This choice of $w$ ensures that each $O(.)$ term in~\eqref{eq:Temp2} is $o(p N_T)$, and hence the coding scheme is capacity-achieving.\end{proof}

Let $n_T$ be equal to the right-hand side of inequality \eqref{eq:DensityBound}. Thus, $Q_L$ fails to include an $n_T\times k$ dense sub-matrix w.p. b.a.b. $\epsilon$.

\begin{lemma}\label{lem:DenseRankProb}Let $Q$ be an $n\times k$ ($k\leq n$) dense matrix over $\mathbb{F}_2$. Then, $\Pr\{\text{rank}(Q)<k\}\leq 2^{-(n-k)}$.\end{lemma}

\begin{proof}The proof can be found in~\cite{HBJ:2011}.\end{proof}

By applying Lemma~\ref{lem:DenseRankProb}, $\Pr\{\text{rank}(Q_L)<k\}$ is b.a.b. $\epsilon$, so long as $k\leq n_T-\log(1/\epsilon)$. By replacing $\epsilon$ with $\dot{\epsilon}$, it follows that the sink node fails to recover all the message vectors w.p. b.a.b. $\epsilon$, so long as $k\leq n_T - \log(1/\epsilon)-1$. In the asymptotic setting, as $N_T$ goes to infinity, $n_T$ can be written as \begin{dmath*}p N_T-(1+o(1))(p N_T L/w+\sqrt{p N_T w\log(wL/\epsilon)}+w\log(wL/\epsilon)).\end{dmath*} We rewrite the last inequality as \begin{eqnarray*}\lefteqn k&\leq& p N_T-(1+o(1))(p N_T L/w \\ && + \sqrt{p N_T w\log(wL/\epsilon)}+w\log(wL/\epsilon))-\log(1/\epsilon)-1.\end{eqnarray*} Let $k_{\text{max}}$ be the largest integer $k$ satisfying this inequality. Thus, $k_{\text{max}}\sim p N_T$, and by replacing $N_T$ with $k/p$ ($N_T\sim k/p$), the following result is immediate.

\begin{theorem}\label{thm:DenseCodesRegularBernoulliActualNon-IdenticalGeneral}The coding delay of a dense code over a line network of $L$ links with deterministic regular transmissions and Bernoulli losses with parameters $\{p_i\}$ is larger than \[\frac{1}{p}\left(k+(1+o(1))\left(\frac{kL}{w}+\sqrt{k\left(w\log\frac{wL}{\epsilon}\right)}+w\log\frac{wL}{\epsilon}\right)\right)\] w.p. b.a.b. $\epsilon$, so long as \[w\log\frac{wL}{\epsilon}=o(k)\] where $w\sim \left(k L^2/\log(k L/\epsilon)\right)^{\frac{1}{3}}$, $p\doteq\min_{1\leq i\leq L}{p_i}$, and the $o(1)$ term goes to $0$ as $k$ goes to infinity.\footnote{In the rest of the theorems, the $o(1)$ term is defined similarly.}\end{theorem}

We now study the average coding delay of dense codes over the traffics with deterministic regular transmissions and Bernoulli losses. In this case, the deviation of the number of packets per partition should not be taken into account. Thus, by replacing $r$ with $\varphi$ in the analysis of the coding delay, the following result can be shown.

\begin{theorem}\label{thm:DenseCodesRegularBernoulliAverageNon-IdenticalGeneral}The average coding delay of a dense code over a network similar to Theorem~\ref{thm:DenseCodesRegularBernoulliActualNon-IdenticalGeneral} is larger than \[\frac{1}{p}\left(k +(1+o(1))\left(\frac{kL}{w}+w\log\frac{wL}{\epsilon}\right)\right)\] w.p. b.a.b. $\epsilon$, so long as \[w\log\frac{wL}{\epsilon}=o(k)\] where $w\sim \left({kL/\log(kL/\epsilon)}\right)^{\frac{1}{2}}$.\end{theorem}

\begin{proof}\renewcommand{\IEEEQED}{}The proof follows the same line as that of Theorem~\ref{thm:DenseCodesRegularBernoulliActualNon-IdenticalGeneral}, except that $r$ needs to be replaced with $\varphi$ in the proof of Lemma~\ref{lem:DensityBound}. Thus, the $O(v\sqrt{\varphi\log(w_T/\epsilon)})$ term in~\eqref{eq:Temp3} and the $O(\sqrt{p N_T w\log(w L/\epsilon)})$ term in~\eqref{eq:Temp2} disappear. Then, it should not be hard to see that the choice of $w$ needs to maximize \begin{dmath}\label{eq:Temp5} p N_T - O(p N_T L/w) - O(w\log(w L/\epsilon)),\end{dmath} instead of~\eqref{eq:Temp2}, subject to condition~\eqref{eq:Temp1}. This can be done by selecting $w$ to be \[\hspace{2.7 in}\sqrt{\frac{p N_T L}{\log(p N_T L/\epsilon)}}.\hspace{2.7 in}\IEEEQEDopen\]\end{proof}

The choice of $w$ in Theorem~\ref{thm:DenseCodesRegularBernoulliAverageNon-IdenticalGeneral} is much larger than the one in Theorem~\ref{thm:DenseCodesRegularBernoulliActualNon-IdenticalGeneral}. This is because, in this case, there is no gap between the lower bound on the number of packet transmissions in each partition and its expectation; and hence, the partitions do not need to be sufficiently long. 


It is worth noting that the preceding results might not provide a very clear picture of how the coding delay or the average coding delay are related to the traffic parameters of the links other than the one(s) with the minimum traffic parameter. However, by applying our analysis technique, while taking into consideration the actual values (and the ordering) of the traffic parameters of the links, new upper bounds (with more details) on the coding delay and the average coding delay can be derived. To be more specific, in such an analysis, for every $1\leq i< L$, depending on whether the $i\textsuperscript{th}$ or the $(i+1)\textsuperscript{th}$ link has a larger traffic parameter, either Lemma~\ref{lem:VerticalT} or Lemma~\ref{lem:HorizontalT} can be used to lower bound the rank of the modified transfer matrix at the $i\textsuperscript{th}$ node, respectively. The rest of the analysis, however, remains the same. For example, the coding delay and the average coding delay of dense codes for the special case with unequal traffic parameters, where no two parameters are equal, can be upper bounded as follows. In particular, the upper bounds, in this case, demonstrate the dependence of the coding delay or the average coding delay on the minimum of the (absolute value of the) difference between the traffic parameters of any two consecutive links in the network.

Let us assume $p_1>p_2>\cdots> p_L$, without loss of generality. Let $p\doteq \min_{1\leq i\leq L} p_i$, $\gamma_e\doteq \min_{1<i\leq L}\gamma_{e_i}$, and $\gamma_{e_i}\doteq |p_i-p_{i-1}|$. Let $r_i\doteq (1-\gamma^{*}_i)\varphi_i$, where $\varphi_i=p_iN_T/w$ and $\gamma^{*}_i\sim\sqrt{(1/\dot{\varphi_i})\log(w_T/\dot{\epsilon})}$. For every $1\leq i\leq L$ and $1\leq j\leq w-L+1$, let $\varphi_{ij}$ be defined as before (i.e., $\varphi_{ij}$ is the number of successful packets in the $j\textsuperscript{th}$ active partition pertaining to the $i\textsuperscript{th}$ link). For all $i,j$, suppose that $\varphi_{ij}$ is larger than or equal to $r_i$, i.e., there exist a sufficiently large number of successful packet transmissions in each partition over each link. (This assumption fails if, for some $1\leq i\leq L$, the number of packets in some active partition over the $i\textsuperscript{th}$ link is less than $r_i$. Hence, the failure occurs w.p. b.a.b. $\dot{\epsilon}$.)

Since all the packet transmissions over the first link are globally dense, for every $1\leq j\leq w-L+1$, $\mathcal{D}(Q_1^j)\geq r_1 j$. Further, by applying Lemma~\ref{lem:HorizontalT}, it can be shown that, for every $1<i\leq L$ and $1\leq j\leq w-L+1$, the inequality $\mathcal{D}(Q_i^j)\geq r_i j$ fails w.p. b.a.b. $ij\dot{\epsilon}/w_T$, so long as \begin{equation}\label{eq:UnequalParameters} {w}\log\frac{w_T}{\epsilon}=o\left(\min\left\{\frac{\gamma_e}{p},1\right\}\cdot p N_T\right).\end{equation}

Let $p$, $\varphi$, $\gamma^{*}$ and $r$ denote $p_L$, $\varphi_L$, $\gamma^{*}_L$ and $r_L$, respectively. Thus, the inequality $\mathcal{D}(Q_L)\geq (1-\gamma^{*})\varphi w_T/L$ fails w.p. b.a.b. ${\epsilon}$. By replacing $\varphi$ with $pN_T/w$, the right-hand side of the last inequality can be written as: \begin{equation}\label{eq:Temp4} p N_T - O(p N_T L/w) -O(\sqrt{p N_T w \log(w L/\epsilon)}).\end{equation} The rest of the analysis is similar to that of Theorem~\ref{thm:DenseCodesRegularBernoulliActualNon-IdenticalGeneral}, except that~\eqref{eq:Temp4} excludes the last term in~\eqref{eq:Temp2}, and the choice of $w$ needs to satisfy condition~\eqref{eq:UnequalParameters}, instead of condition~\eqref{eq:Temp1}.

\begin{theorem}\label{thm:DenseCodesRegularBernoulliActualNon-Identical}Consider a sequence of unequal parameters $\{p_i\}_{1\leq i\leq L}$. The coding delay of a dense code over a line network of $L$ links with deterministic regular transmissions and Bernoulli losses with parameters $\{p_i\}$ is larger than \[\frac{1}{p}\left(k+(1+o(1))\left(\frac{kL}{w}+\sqrt{k\left(w\log\frac{wL}{\epsilon}\right)}\right)\right)\] w.p. b.a.b. $\epsilon$, so long as \[{w}\log\frac{wL}{\epsilon}=o\left(\min\left\{\frac{\gamma_e}{p},1\right\}\cdot k\right)\] where $w\sim\left(k L^2/\log(k L/\epsilon)\right)^{\frac{1}{3}}$, $p \doteq \min_{1\leq i\leq L}p_i$, $\gamma_{e}\doteq \min_{1< i\leq L} \gamma_{e_i}$, and $\gamma_{e_i} \doteq |p_i-p_{i-1}|$.\end{theorem}

In the case of the average coding delay, the analysis follows the same line as that of Theorem~\ref{thm:DenseCodesRegularBernoulliAverageNon-IdenticalGeneral}, except that the choice of $w$ needs to maximize \begin{equation}\label{eq:Temp6} pN_T - O(p N_T L/w)\end{equation} subject to condition~\eqref{eq:UnequalParameters}, instead of~\eqref{eq:Temp5} subject to condition~\eqref{eq:Temp1}.

\begin{theorem}\label{thm:DenseCodesRegularBernoulliAverageNon-Identical}The average coding delay of a dense code over a network similar to Theorem~\ref{thm:DenseCodesRegularBernoulliActualNon-Identical} is larger than \[\frac{1}{p}\left(k +(1+o(1))\left(\frac{k L}{w}\right)\right)\] w.p. b.a.b. $\epsilon$, so long as \[{w}\log\frac{wL}{\epsilon}=o\left(\min\left\{\frac{\gamma_e}{p},1\right\}\cdot k\right),\] i.e., $w\sim k/(f(k)\log(k L/\epsilon))$, and $f(k)$ goes to infinity, as $k$ goes to infinity, such that $f(k)=o\left(k/(L\log(kL/\epsilon))\right)$.\end{theorem}

\subsection{Chunked Codes}
In a chunked coding scheme, the set of $k$ message vectors at the source node is divided into $q$ disjoint subsets, called \emph{chunks}, each of size $\alpha=k/q$. The source node, at each transmission time, chooses a chunk independently at random, and transmits a packet by randomly linearly combining the message vectors belonging to the underlying chunk.\footnote{The ``random'' scheduling of the chunks and the ``random'' coding within the chunks have been shown to perform effectively when the feedback information is not available at the network nodes \cite{HBJ:2011}. However, for cases with feedback, more efficient scheduling policies and coding schemes have been proposed in the literature (e.g., see \cite{HBJ2:2012}).} Each non-source non-sink node, at the time of each transmission, chooses a chunk independently at random, and transmits a packet by randomly linearly combining its previously received packets pertaining to the underlying chunk. The sink node can decode a chunk, so long as it receives an innovative collection of packets pertaining to the underlying chunk of a size equal to the size of the chunk.

\subsubsection{Capacity-Achieving Scenarios}\label{subsubsec:CCCACH}
In a CC, at each transmission time, a chunk is chosen w.p. $1/q$, and a packet transmission over the $i\textsuperscript{th}$ link is successful w.p. $p_i$. Thus the probability that a given packet transmission over the $i\textsuperscript{th}$ link is successful and pertains to a given chunk is $p_i/q$. Thus by replacing $p_i$ with $p_i/q$ in the analysis of dense codes in Section~\ref{subsec:DC}, the coding delay and the average coding delay of CC in a capacity-achieving scenario will be upper bounded.

The results of dense codes are indeed a special case of those of CC with one chunk of size $k$. It is, however, worth noting that, due to the change in the parameters, the number of partitions $w$ needs to satisfy a new condition: $wq\log\frac{w_T q}{\epsilon}=o(p N_T)$ or $wq\log\frac{w_T q}{\epsilon}=o\left(\min\left\{\frac{\gamma_e}{p},1\right\} \cdot p N_T\right)$, instead of condition~\eqref{eq:Temp1} or~\eqref{eq:UnequalParameters}, in the proofs of Theorems~\ref{thm:CapAchCodDelGeneral} and~\ref{thm:CapAchAveCodDelGeneral}, or those of Theorems~\ref{thm:CapAchCodDelSpecial} and~\ref{thm:CapAchAveCodDelSpecial}, respectively. Further, by replacing $w$ with its optimal choice in the new version of~\eqref{eq:Temp2},~\eqref{eq:Temp5},~\eqref{eq:Temp4} and~\eqref{eq:Temp6}, each $O(.)$ term needs to be $o(pN_T/q)$ in order to ensure that CC are capacity-achieving in the underlying case. Such a condition lower bounds the size of chunks ($\alpha$) by a function super-logarithmic in the message size ($k$).

\begin{theorem}\label{thm:CapAchCodDelGeneral}The coding delay of a CC with $q$ chunks over a line network of $L$ links with deterministic regular transmissions and Bernoulli losses with parameters $\{p_i\}$ is larger than \begin{dmath*}\frac{1}{p}\left(k+(1+o(1))\left(\frac{k L}{w}+\sqrt{k\left(wq\log\frac{w q L}{\epsilon}\right)}+wq\log\frac{w q L}{\epsilon}\right)\right)\end{dmath*} w.p. b.a.b. $\epsilon$, so long as \[q=o({k}/({L\log(kL/\epsilon)})),\] and \[wq\log\frac{wLq}{\epsilon}=o(k)\] where $w\sim\left(kL^2/(q\log(kL/\epsilon))\right)^{\frac{1}{3}}$, and $p\doteq \min_{1\leq i\leq L}p_i$.\end{theorem}

\begin{proof}\renewcommand{\IEEEQED}{}The proof follows the same line as that of Theorem~\ref{thm:DenseCodesRegularBernoulliActualNon-IdenticalGeneral} by implementing the following modifications. Let us replace $p$ and $\epsilon$ with $p/q$ and $\epsilon/q$, respectively. Then, $\varphi=pN_T/wq$, and $r=(1-\gamma^{*})\varphi$, where $\gamma^{*}\sim\sqrt{(1/\dot{\varphi})\ln(w_T q/\dot{\epsilon})}$. For every $1\leq i\leq L$, and $1\leq j\leq w-L+1$, let $\mathcal{D}(Q_i^j)$, $\mathcal{D}_p(Q_i^j)$, and $r_{ij}$ be defined as in Section~\ref{subsec:DC}, but only restricted to the packets pertaining to a given chunk (not all the chunks). For every $i,j$, $\mathcal{D}(Q_i^j)$ can be lower bounded as follows (the proofs are very similar to those of Lemmas~\ref{lem:Omegaii} and~\ref{lem:Omegaij}): for every $1\leq j\leq w-L+1$, $\mathcal{D}(Q_1^j)\geq rj$, and for every $1<i\leq L$ and $1\leq j\leq w-L+1$, $\mathcal{D}(Q_i^j)$ fails to be larger than $rj-j(1+o(1))\log(w_T q/\epsilon)$, w.p. b.a.b. $ij\dot{\epsilon}/w_T q$, so long as \begin{equation}\label{eq:Temp15} wq\log\frac{w_T q}{\epsilon}=o(p N_T).\end{equation} Thus the number of dense packets pertaining to a given chunk at the sink node fails to be larger than \begin{eqnarray}\label{eq:Temp16}
 \lefteqn{\frac{p N_T}{q} - O\left(\frac{p N_T L}{wq}\right) -  } \nonumber\\
   && O\left(\sqrt{\frac{p N_T w}{q} \log\frac{w q L}{\epsilon}}\right)-O\left(w\log\frac{w q L}{\epsilon}\right)
\end{eqnarray} w.p. b.a.b. $\epsilon/q$. In order to maximize~\eqref{eq:Temp16} subject to condition~\eqref{eq:Temp15}, we select $w$ to be \[\sqrt[3]{\frac{p N_T L^2}{q\log(p N_T L/\epsilon)}}.\] Now let us assume that $N_T$ is $(1+o(1))k/p$. By replacing $\epsilon$ with $\dot{\epsilon}$, in the preceding results, and by replacing $k$ and $\epsilon$ with $k/q$ and $\dot{\epsilon}/q$, respectively, in Lemma~\ref{lem:DenseRankProb}, it follows that the sink node fails to decode a given chunk w.p. b.a.b. $\epsilon/q$, so long as $N_T$ is larger than \begin{dmath}\label{eq:Temp17} \frac{1}{p}\left(k+(1+o(1))\left(\frac{k L}{w}+\sqrt{k\left(wq\log\frac{w q L}{\epsilon}\right)}+wq\log\frac{w q L}{\epsilon}\right)\right).\end{dmath} Taking a union bound over all the chunks, it follows that the sink node fails to decode all the chunks w.p. b.a.b. $\epsilon$, so long as $N_T$ is larger than~\eqref{eq:Temp17}. To ensure that the lower bound on $N_T$ is $(1+o(1))k/p$, all the terms in~\eqref{eq:Temp17}, excluding the first one, need to be $o(k/p)$. This condition is met so long as $q$ is \[\hspace{2.65 in}o\left(\frac{k}{L\log(kL/\epsilon)}\right).\hspace{2.65 in}\IEEEQEDopen\]\end{proof}

\begin{theorem}\label{thm:CapAchAveCodDelGeneral}The average coding delay of a CC with $q$ chunks over a network similar to Theorem~\ref{thm:CapAchCodDelGeneral} is larger than \begin{dmath*}\frac{1}{p}\left(k+(1+o(1))\left(\frac{k L}{w}+wq\log\frac{w q L}{\epsilon}\right)\right)\end{dmath*} w.p. b.a.b. $\epsilon$, so long as \[q=o({k}/({L\log(kL/\epsilon)})),\] and \[wq\log\frac{wLq}{\epsilon}=o(k)\] where $w\sim\left(kL/(q\log(kL/\epsilon))\right)^{\frac{1}{2}}$.\end{theorem}

\begin{proof}The proof is similar to that of Theorem~\ref{thm:CapAchCodDelGeneral}, except that $r$ needs to be replaced with $\varphi$. This implies that the third term in~\eqref{eq:Temp16} disappears. Thus, by selecting $w$ to be \[\sqrt{\frac{p N_T L}{q\log(p N_T L/\epsilon)}}\] in order to maximize a new version of~\eqref{eq:Temp16} (i.e., where the third term in~\eqref{eq:Temp16} is excluded), subject to condition~\eqref{eq:Temp15}, it follows that the sink node fails to decode all the chunks w.p. b.a.b. $\epsilon$, so long as $N_T$ is larger than \begin{dmath}\label{eq:Temp18} \frac{1}{p}\left(k+(1+o(1))\left(\frac{k L}{w}+wq\log\frac{w q L}{\epsilon}\right)\right).\end{dmath} The rest of the proof follows that of Theorem~\ref{thm:CapAchCodDelGeneral}.\end{proof}

In the case of unequal traffic parameters, the coding delay and the average coding delay are upper bounded as follows.

\begin{theorem}\label{thm:CapAchCodDelSpecial}The coding delay of a CC with $q$ chunks over a line network of $L$ links with deterministic regular transmissions and Bernoulli losses with unequal parameters $\{p_i\}$ is larger~than \begin{dmath*}\frac{1}{p}\left(k+(1+o(1))\left(\frac{k L}{w}+\sqrt{k\left(wq\log\frac{w q L}{\epsilon}\right)}\right)\right)\end{dmath*} w.p. b.a.b. $\epsilon$, so long as \[q=o\left(\min\left\{\frac{\gamma_e}{p},1\right\}\cdot {k}/({L\log(kL/\epsilon)})\right),\] where $w\sim\left(kL^2/(q\log(kL/\epsilon))\right)^{\frac{1}{3}}$, $p\doteq \min_{1\leq i\leq L}p_i$, $\gamma_e\doteq\min_{1<i\leq L} \gamma_{e_i}$, and $\gamma_{e_i}\doteq |p_i-p_{i-1}|$.\end{theorem}

\begin{proof}\renewcommand{\IEEEQED}{}By replacing $p$ and $\epsilon$ with $p/q$ and $\epsilon/q$, respectively, in the proof of Theorem~\ref{thm:DenseCodesRegularBernoulliActualNon-Identical}, it follows that the number of dense packets pertaining to a given chunk at the sink node fails to be larger than \begin{eqnarray}\label{eq:Temp19}
 \lefteqn{\frac{p N_T}{q} - O\left(\frac{p N_T L}{wq}\right) -  } \nonumber\\
   && O\left(\sqrt{\frac{p N_T w}{q} \log\frac{w q L}{\epsilon}}\right)
\end{eqnarray} w.p. b.a.b. $\epsilon/q$, so long as \begin{equation}\label{eq:Temp20} wq\log\frac{w_T q}{\epsilon}=o\left(\min\left\{\frac{\gamma_e}{p},1\right\}\cdot p N_T\right).\end{equation} The rest of the proof is similar to that of Theorem~\ref{thm:CapAchCodDelGeneral}, except that~\eqref{eq:Temp19} excludes the last term in~\eqref{eq:Temp16}, and the choice of $w$ needs to satisfy condition~\eqref{eq:Temp20}, instead of condition~\eqref{eq:Temp15}. By selecting $w$ to be \[\sqrt[3]{\frac{p N_T L^2}{q\log(p N_T L/\epsilon)}}\] in order to maximize~\eqref{eq:Temp19} subject to condition~\eqref{eq:Temp20}, it follows that the sink node fails to decode all the chunks w.p. b.a.b. $\epsilon$, so long as $N_T$ is larger than \begin{dmath}\label{eq:Temp21} \frac{1}{p}\left(k+(1+o(1))\left(\frac{k L}{w}+\sqrt{k\left(wq\log\frac{w q L}{\epsilon}\right)}\right)\right).\end{dmath} In~\eqref{eq:Temp21}, each term, except the largest one, needs to be $o(k/p)$, and this condition is met so long as $q$ is \[\hspace{2.22 in}o\left(\min\left\{\frac{\gamma_e}{p},1\right\}\cdot\frac{k}{L\log(kL/\epsilon)}\right).\hspace{2.22 in}\IEEEQEDopen\]\end{proof}

\begin{theorem}\label{thm:CapAchAveCodDelSpecial}The average coding delay of a CC with $q$ chunks over a network similar to Theorem~\ref{thm:CapAchCodDelSpecial} is larger than \begin{dmath*}\frac{1}{p}\left(k+(1+o(1))\left(\frac{k L}{w}\right)\right)\end{dmath*} w.p. b.a.b. $\epsilon$, so long as \[q=o\left(\min\left\{\frac{\gamma_e}{p},1\right\}\cdot {k}/(f(k){L\log(kL/\epsilon)})\right),\] where $w\sim k/(q f(k)\log(kL/\epsilon))$, and $f(k)$ goes to infinity, as $k$ goes to infinity, such that $f(k)=o\left(k/(L\log(kL/\epsilon))\right)$.\end{theorem}

\begin{proof}\renewcommand{\IEEEQED}{}The proof follows the same line as that of Theorem~\ref{thm:CapAchCodDelGeneral}, except that the choice of $w$ needs to maximize \begin{equation}\label{eq:Temp22} \frac{p N_T}{q} - O\left(\frac{p N_T L}{wq}\right)\end{equation} subject to condition~\eqref{eq:Temp20}. To do so, we select $w$ to be \[\frac{p N_T}{q f(p N_T)\log(p N_T L/\epsilon)},\] where $f(n)$ goes to infinity, as $n$ goes to infinity, such that $f(n)=o\left(n/(L\log(n L/\epsilon))\right)$. The sink node fails to decode all the chunks w.p. b.a.b. $\epsilon$, so long as $N_T$ is larger than \begin{dmath}\label{eq:Temp23} \frac{1}{p}\left(k+(1+o(1))\left(\frac{k L}{w}\right)\right).\end{dmath} The second term in~\eqref{eq:Temp23} needs to be $o(k/p)$, and this condition is met so long as $q$ is \[\hspace{2.1 in}o\left(\min\left\{\frac{\gamma_e}{p},1\right\}\cdot\frac{k}{f(k)L\log(kL/\epsilon)}\right).\hspace{2.1 in}\IEEEQEDopen\]\end{proof}

\subsubsection{Capacity-Approaching-with-a-Gap Scenarios}\label{subsubsec:CCCAPP}
By the results of Section~\ref{subsubsec:CCCACH}, one can conclude that CC are not capacity-achieving if the size of the chunks does not comply with condition $\alpha=\omega({L\log(kL/\epsilon)})$. Also, the analysis of Section~\ref{subsec:DC} does not apply to CC with chunks of small sizes violating the above condition. From a computational complexity perspective, CC with chunks of smaller sizes are, however, of more practical interest (e.g., linear-time CC with constant-size chunks). In the following,~we study CC with chunks of a size constant in the message~size.

Let $\{p_i\}_{1\leq i\leq L}$ be an arbitrary sequence of traffic parameters, and let $p\doteq \min_{1\leq i\leq L}p_i$. Let the size of the chunks $\alpha$ ($=k/q$) be a constant in the message size $k$, i.e., $\alpha=O(1)$. Let the time interval $(0,N_T]$ and its $w$ disjoint partitions be defined as in~Section~\ref{subsec:DC}. Let $\varphi_{ij}$ be the number of packets (pertaining to a given chunk) in the partition $I_{ij}$, and $\varphi_i$ be the expected value of $\varphi_{ij}$. Let $\varphi\doteq \min_{1\leq i\leq L}\varphi_i$. Then, $\varphi_i=p_i N_T/wq$, and $\varphi=pN_T/wq$. Let $N_T=(1+\gamma_c)k/p$, where $0<\gamma_c<1$ is an arbitrarily small constant. By replacing $N_T$ with $(1+\gamma_c)k/p$, it follows that $\varphi=(1+\gamma_c)\alpha/w$. Further, it is not hard to see that $\varphi=O(1)$, since $w$ has to be a constant (otherwise, if $w$ goes to infinity, as $N_T$ goes to infinity, then $\varphi$ goes to $0$, and for such a case, our analysis is not valid).

By applying the Chernoff bound, it can be shown that $\Pr\{\varphi_{ij}<(1-\gamma^{*})\varphi\}\leq e^{-{\gamma^{*}}^2\dot{\varphi}}$, for every $0<\gamma^{*}<1$. Taking $e^{-{\gamma^{*}}^2\dot{\varphi}}\leq \dot{\gamma_b}/w_T$, it follows that $\varphi_{ij}$ is not larger than or equal to $r\doteq(1-\gamma^{*})\varphi$ w.p. b.a.b. $\dot{\gamma_b}/w_T$, where $\gamma^{*}$ is chosen to be the smallest real number larger than or equal to $\sqrt{(1/\dot{\varphi})\ln(w_T/\dot{\gamma_b})}$ such that $r$ ($=(1-\gamma^{*})\varphi$) is an integer. It is not hard to see that $\gamma^{*}=O(1)$. Taking a union bound over all the active partitions of all links, it follows that $\varphi_{ij}$ is not larger than or equal to $r$ w.p. b.a.b. $\dot{\gamma_b}$.


Let $\mathcal{D}(Q_{i}^j)$ be the number of dense packets pertaining to a given chunk in the first $j$ active partitions over the $i\textsuperscript{th}$ link.

By applying Lemma~\ref{lem:HorizontalT}, it can be shown that: (i) for every $1\leq j\leq w-L+1$, $\mathcal{D}(Q_1^j)\geq rj$, (ii) for every $1<i\leq L$, the inequality $\mathcal{D}(Q_i^1)\geq r-\log(w_T/\dot{\gamma_b})$ fails w.p. b.a.b. $i\dot{\gamma_b}/w_T$, and (iii) for every $1<i\leq L$ and $1<j\leq w-L+1$, the inequality $\mathcal{D}(Q_i^j)\geq r-j\log(w_T/\dot{\gamma_b})-\log((j+1)w_T/\dot{\gamma_b})$ fails w.p. b.a.b. $ij\dot{\gamma_b}/w_T$, so long as \begin{equation}\label{eq:Temp10}\alpha=\Omega \left(w^2\log\frac{w_T}{{\gamma_b}}\right).\end{equation}

By using the above results, it follows that the number of dense packets pertaining to a given chunk at the sink node fails to be lower bounded by \begin{dmath}\label{eq:Temp11} \frac{w_T\varphi}{L} - O\left(\frac{w_T}{L}\sqrt{\varphi\log\frac{w_T}{\gamma_b}}\right)-O\left(\frac{w_T}{L}\log\frac{w_T}{\gamma_b}\right)\end{dmath} w.p. b.a.b. $\gamma_b$. The lower bound is non-negative so long as $\alpha=\Omega\left(w\log({w_T}/{\gamma_b})\right)$, and this condition holds so long as condition~\eqref{eq:Temp10} holds. We select $w$ to be $\sqrt[3]{\alpha L^2/\log(\alpha L/\gamma_b)}$ to maximize~\eqref{eq:Temp11}. By replacing $w$ with this value,~\eqref{eq:Temp10} can be rewritten as \begin{equation}\label{eq:Temp12} \alpha=\Omega\left(L^4\log\frac{L}{\gamma_b}\right).\end{equation}

By replacing $\gamma_b$ with $\dot{\gamma_b}$, and by applying Lemma~\ref{lem:DenseRankProb}, it follows that the sink node fails to decode a given chunk w.p. b.a.b. $\gamma_b$, so long as~\eqref{eq:Temp11} is larger than $\alpha+\log({1}/{\dot{\gamma_b}})$. By replacing our choice of $w$ in~\eqref{eq:Temp11}, it can be seen that, excluding the first term, the second term dominates the rest. By replacing $\varphi$ with $(1+\gamma_c)\alpha/w$, the decoding condition becomes \begin{equation}\label{eq:Temp13}\alpha=\Omega\left(\frac{L}{\gamma^3_c}\log\frac{L}{\gamma_b\gamma_c}\right).\end{equation} Thus, a given chunk fails to be decodable w.p. b.a.b. $\gamma_b$ so long as both conditions~\eqref{eq:Temp12} and~\eqref{eq:Temp13} are met. In other words, the expected fraction of undecodable chunks is bounded from above by $\gamma_b$. By using a martingale argument similar to the one in~\cite{HBJ:2011} (by constructing a martingale sequence over the number of undecodable chunks), the concentration of the fraction of undecodable chunks around the expectation can be shown as follows. The proof is omitted to avoid repetition.

\begin{lemma}\label{lem:Concentration}By applying a CC with chunks of size $\alpha$, satisfying both conditions~\eqref{eq:Temp12} and~\eqref{eq:Temp13}, the fraction of undecodable chunks at the sink node until time $N_T = (1+\gamma_c)k/p$ is larger than $(1+\gamma_a)\gamma_b$, w.p. b.a.b. $\epsilon$, so long as \begin{equation}\label{eq:Temp14}{\alpha^2}/{\gamma^2_a\gamma^2_b}=o({k}/{\log({1}/{\epsilon})}),\end{equation} where $0<\gamma_a,\gamma_b,\gamma_c<1$ are arbitrary constants.\end{lemma}

By the result of Lemma~\ref{lem:Concentration}, the fraction of chunks which are not decodable until time $N_T$ becomes larger than $(1+\gamma_a)\gamma_b$, w.p. b.a.b. $\epsilon$. Since $\gamma_a,\gamma_b$ are non-zero constants, a CC, alone, might not decode all the chunks. However, the completion of decoding of all the chunks is guaranteed by devising a proper precoding scheme~\cite{HBJ:2011}. The precoding works as follows: The set of $k$ message vectors at the source node constitute the input of a capacity-achieving erasure code, called \emph{precode}. The rate of the precode is $1-(1+\gamma_a)\gamma_b$, i.e., the precode decoder can correct up to a fraction $(1+\gamma_a)\gamma_b$ of erasures,\footnote{The precode does not have to be capacity-achieving and its rate can be arbitrarily close to $1-(1+\gamma_a)\gamma_b$, yet, it has to be able to correct up to a fraction $(1+\gamma_a)\gamma_b$ of erasures (for more details, see~\cite{HBJ:2011}).} and the number of the coded packets at the output of the precode, called \emph{intermediate packets}, is $\left(1+(1+\gamma_a)\gamma_b+O(\gamma^2_b)\right)k$. By applying a CC with chunks of size $\alpha$, satisfying conditions~\eqref{eq:Temp12},~\eqref{eq:Temp13} and~\eqref{eq:Temp14}, the fraction of the intermediate packets that are not recoverable at the output of the CC decoder until time $(1+\gamma_c)\left(1+(1+\gamma_a)\gamma_b+O(\gamma^2_b)\right)\frac{k}{p}$ is larger than $(1+\gamma_a)\gamma_b$, w.p. b.a.b. $\epsilon$. Then, the precode decoder can recover all the $k$ message vectors from the set of recovered intermediate packets. Therefore, the coding delay of a CC with precoding (CCP) is upper bounded as follows.

\begin{theorem}\label{thm:CapAppCodDelGeneral}The coding delay of a CCP with chunks of size $\alpha$ and a capacity-achieving erasure code of rate $1-(1+\gamma_a)\gamma_b$, over a line network of $L$ links with deterministic regular transmissions and Bernoulli losses with parameters $\{p_i\}$ is larger than $(1+\gamma_c)\left(1+(1+\gamma_a)\gamma_b+O(\gamma^2_b)\right)\frac{k}{p}$, w.p. b.a.b.~$\epsilon$, so long~as \[\alpha=\Omega\left(\left\{\left(\frac{L}{\gamma^3_{c}}\log\frac{L}{\gamma_b\gamma_c}\right),\left(L^4 \log \frac{L}{\gamma_b}\right)\right\}\right),\] and $\alpha^2/\gamma^2_a\gamma^2_b=o(k/\log(1/\epsilon))$, where $0<\gamma_a,\gamma_b,\gamma_c<1$ are arbitrary constants, and $p\doteq \min_{1\leq i\leq L}p_i$.\end{theorem}

In the case of the average coding delay of a CC with precoding, the following can be shown similar to Theorem~\ref{thm:CapAppCodDelGeneral} by replacing $r$ with $\varphi$.

\begin{theorem}\label{thm:CapAppAveCodDelGeneral}The average coding delay of a CCP with chunks of size $\alpha$ and a capacity-achieving erasure code of rate $1-(1+\gamma_a)\gamma_b$, over a network similar to Theorem~\ref{thm:CapAppCodDelGeneral} is larger than $(1+\gamma_c)\left(1+(1+\gamma_a)\gamma_b+O(\gamma^2_b)\right)\frac{k}{p}$, w.p. b.a.b. $\epsilon$, so long as \[\alpha=\Omega\left(\frac{L}{\gamma_{c}}\log\frac{L}{\gamma_b \gamma_c}\right),\] and $\alpha^2/\gamma^2_a\gamma^2_b=o(k/\log(1/\epsilon))$, where $0<\gamma_a,\gamma_b,\gamma_{c}<1$ are arbitrary constants.\end{theorem}

In the special case of unequal traffic parameters, the coding delay and the average coding delay of CC with precoding can be upper bounded as follows. The proofs follow the same line as in the general case except that a new set of conditions needs to be satisfied based on the assumption that no two traffic parameters are equal.

\begin{theorem}\label{thm:CapAppCodDelSpecial}The coding delay of a CCP with chunks of size $\alpha$ and a capacity-achieving erasure code of rate $1-(1+\gamma_a)\gamma_b$, over a line network of $L$ links with deterministic regular transmissions and Bernoulli losses with unequal parameters $\{p_i\}$ is larger than $(1+\gamma_c)\left(1+(1+\gamma_a)\gamma_b+O(\gamma^2_b)\right)\frac{k}{p}$, w.p. b.a.b. $\epsilon$, so long as \[\alpha=\Omega\left(\left\{\left(\frac{L}{\gamma^3_{c}}\log\frac{L}{\gamma_b\gamma_c}\right),\left(\frac{L}{\gamma^3_{e}} \log \frac{L}{\gamma_e\gamma_b}\right)\right\}\right),\] and $\alpha^2/\gamma^2_a\gamma^2_b=o(k/\log(1/\epsilon))$, where $0<\gamma_a,\gamma_b,\gamma_c<1$ are arbitrary constants, $p\doteq \min_{1\leq i\leq L}p_i$, $\gamma_e\doteq\min_{1<i\leq L} \gamma_{e_i}$, and $\gamma_{e_i}\doteq |p_i-p_{i-1}|$.\end{theorem}

\begin{proof}Let us assume $p_1>p_2>\cdots> p_L$, without loss of generality. Let $p\doteq \min_{1\leq i\leq L} p_i$, $\gamma_e\doteq \min_{1<i\leq L}\gamma_{e_i}$, and $\gamma_{e_i}\doteq |p_i-p_{i-1}|$. Let $r_i\doteq (1-\gamma^{*}_i)\varphi_i$, where $\varphi_i=p_iN_T/wq$ and $\gamma^{*}_i\sim\sqrt{(1/\dot{\varphi_i})\log(w_T/\dot{\gamma_b})}$, and $0<\gamma_b<1$ is an arbitrary constant. For every $1\leq i\leq L$ and $1\leq j\leq w-L+1$, let $\varphi_{ij}$ be the number of packets (pertaining to a given chunk) in the partition $I_{ij}$ (the $j\textsuperscript{th}$ partition pertaining to the $i\textsuperscript{th}$ link), where the time interval $(0,N_T]$ is split into $w$ partitions of length $N_T/w$, and let $\varphi_i$ be the expected value of $\varphi_{ij}$. For all $i,j$, suppose that $\varphi_{ij}$ is larger than or equal to $r_i$. Let $N_T=(1+\gamma_c)k/p$, where $0<\gamma_c<1$ is an arbitrarily small constant. By replacing $N_T$ with $(1+\gamma_c)k/p$, it follows that $\varphi_i=(1+\gamma_c)p_i\alpha/pw$, and $\varphi=O(1)$, similar to those in the proof of Theorem~\ref{thm:CapAppCodDelGeneral}.

For every $1\leq i\leq L$ and $1\leq j\leq w-L+1$, let $\mathcal{D}(Q_i^j)$, $\mathcal{D}_p(Q_i^j)$, and $r_{ij}$ be defined as in the proof of Theorem~\ref{thm:CapAchCodDelGeneral}. For every $1\leq j\leq w-L+1$, $\mathcal{D}(Q_1^j)\geq r_1 j$ (since all the packets pertaining to any chunk over the first link are globally dense). For every $1<i\leq L$ and $1\leq j\leq w-L+1$, by applying Lemma~\ref{lem:HorizontalT}, it can be shown that the inequality $\mathcal{D}(Q_i^j)\geq r_i j$ fails w.p. b.a.b. $ij\dot{\gamma_b}/w_T$, so long as \begin{equation}\label{eq:Temp24}\alpha=\Omega\left(\frac{w}{\gamma_e^2}\log\frac{w_T}{\gamma_b}\right).\end{equation}

Let $\varphi$, $\gamma^{*}$ and $r$ denote $\varphi_L$, $\gamma^{*}_L$ and $r_L$, respectively. Thus, the number of dense packets pertaining to a given chunk at the sink node fails to be larger than \begin{dmath}\label{eq:Temp25}(1+\gamma_c)\alpha-O\left(\frac{\alpha L}{w}\right)-O\left(\sqrt{\alpha w\log\frac{w_T}{\gamma_b}}\right).\end{dmath} We select $w$ to be \[\sqrt[3]{\frac{\alpha L^2}{\log(w_T/\gamma_b)}}\] to maximize~\eqref{eq:Temp25} subject to condition~\eqref{eq:Temp24}. For this choice of $w$, condition~\eqref{eq:Temp24} is met so long as \begin{equation}\label{eq:Temp26}\alpha=\Omega\left(\frac{L}{\gamma^3_e}\log\frac{L}{\gamma_e\gamma_b}\right).\end{equation} By replacing $\gamma_b$ with $\dot{\gamma_b}$ in the preceding results, and substituting the selected value of $w$ in~\eqref{eq:Temp25}, the result of Lemma~\ref{lem:DenseRankProb} shows that the sink node fails to decode a given chunk w.p. b.a.b. $\gamma_b$, so long as~\eqref{eq:Temp25} is larger than $\alpha+\log(1/\dot{\gamma_b})$. Based on the properties of the notation $\Omega(.)$, the latter condition is met so long as \begin{equation}\label{eq:Temp27}\alpha=\Omega\left(\frac{L}{\gamma^3_c}\log\frac{L}{\gamma_b\gamma_c}\right).\end{equation}
The rest of the proof is similar to the proof of Theorem~\ref{thm:CapAppCodDelGeneral}, except that in this case conditions~\eqref{eq:Temp26} and~\eqref{eq:Temp27} need to be met, instead of conditions~\eqref{eq:Temp12} and~\eqref{eq:Temp13}.\end{proof}

\begin{theorem}\label{thm:CapAppAveCodDelSpecial}The average coding delay of a CCP with chunks of size $\alpha$ and a capacity-achieving erasure code of rate $1-(1+\gamma_a)\gamma_b$, over a network similar to Theorem~\ref{thm:CapAppCodDelSpecial} is larger than $(1+\gamma_c)\left(1+(1+\gamma_a)\gamma_b+O(\gamma^2_b)\right)\frac{k}{p}$, w.p. b.a.b. $\epsilon$, so long as \[\alpha=\Omega\left(\frac{L}{\gamma^2_e\gamma_{c}}\log\frac{L}{\gamma_b \gamma_c}\right),\] and $\alpha^2/\gamma^2_a\gamma^2_b=o(k/\log(1/\epsilon))$, where $0<\gamma_a,\gamma_b,\gamma_{c}<1$ are arbitrary constants.\end{theorem}

\begin{proof}\renewcommand{\IEEEQED}{}The proof follows the same line as that of Theorem~\ref{thm:CapAppCodDelSpecial}, except that the choice of $w$ needs to maximize \begin{equation}\label{eq:Temp28} (1+\gamma_c)\alpha-O\left(\frac{\alpha L}{w}\right)\end{equation} subject to condition~\eqref{eq:Temp24}. To do so, the choice of $w$ needs to be $\Omega(L/\gamma_c)$, and hence, condition~\eqref{eq:Temp24} becomes \[\hspace{2.45 in}\alpha=\Omega\left(\frac{L}{\gamma^2_e\gamma_c}\log\frac{L}{\gamma_b\gamma_c}\right).\hspace{2.45 in}\IEEEQEDopen\]
\end{proof}

\section{Poisson Transmissions and Bernoulli Losses}\label{sec:PoissonTraffic}
In the case of Bernoulli losses and Poisson transmissions with parameters $\{p_i\}_{1\leq i\leq L}$ and $\{\lambda_i\}_{1\leq i\leq L}$, the points in time at which the arrivals/departures occur over the $i\textsuperscript{th}$ link follow a Poisson process with parameter $\lambda_i p_i$. Thus the number of packets pertaining to a given chunk (note that a dense code is a CC with only one chunk), in each partition pertaining to the $i\textsuperscript{th}$ link, has a Poisson distribution with the expected value $\lambda_i p_i N_T/wq$. Since the result of Chernoff bound also holds for Poisson random variables (see \cite[Theorem~A.1.15]{AS:2008}), the main results in Section~\ref{sec:BernoulliLossRegularTraffic} apply to this case by replacing $p$ with $\min_{1\leq i\leq L}\lambda_i p_i$.

\section{Discussion}\label{sec:Discussion}
\subsection{Dense Codes}\label{subsec:DiscussionDenseCodes}
The upper bounds on the coding delay and the average coding delay, derived in this paper, are valid for any arbitrary choice of $\epsilon$. However, in the following, to compare our results with those of~\cite{PFS:2005} and~\cite{DDHE:2009}, we focus on the case where $\epsilon$ goes to $0$ polynomially fast, as $k$ goes to infinity (i.e., $\epsilon=1/k^c$, for some constant $c>0$). For such a choice of $\epsilon$, the upper bounds on the coding delay and the average coding delay hold w.p. $1$, as $k$ goes to infinity.

In~\cite{PFS:2005}, the average coding delay of dense codes over the networks of length $2$ with deterministic regular transmissions and Bernoulli losses with equal parameters ($p$) is shown to be upper bounded by $\frac{1}{p}(k+O(\sqrt{k\log k}))$. The result of Theorem~\ref{thm:DenseCodesRegularBernoulliAverageNon-IdenticalGeneral} indicates that the average coding delay of dense codes over the networks of length $L$ with similar traffics as above (i.e., the special case of identical links with equal parameters)\footnote{One should note that Theorems~\ref{thm:DenseCodesRegularBernoulliActualNon-IdenticalGeneral} and~\ref{thm:DenseCodesRegularBernoulliAverageNon-IdenticalGeneral} are not restricted to the special case of identical links, and hold true for any arbitrary sequence of parameters.} is upper bounded by $\frac{1}{p}(k+(1+o(1))(\sqrt{kL\log(kL)}))$. This is consistent with the result of~\cite{PFS:2005}, although the bound presented here provides more details on the smaller terms in the $O(.)$ term.

The result of Theorem~\ref{thm:DenseCodesRegularBernoulliActualNon-IdenticalGeneral} suggests that the coding delay of dense codes over network scenarios as above is upper bounded by $\frac{1}{p}(k+(1+o(1))({k^2 L \log(kL)})^{\frac{1}{3}})$. One should note that there has been no result on the coding delay of dense codes over identical links in the existing literature. In fact, this was posed as an open problem in \cite{DDHE:2009}. It is also noteworthy that unlike the analysis of~\cite{DDHE:2009}, our analysis does not rely on the existence of a single worst link, and hence is applicable to the special case of identical links.

In~\cite{DDHE:2009}, the average coding delay of dense codes over the networks of length $L$ with deterministic regular transmissions and Bernoulli losses with parameters $\{p_i\}$ was upper bounded by $\frac{k}{p}+\sum_{i\neq \nu} \frac{1-p}{p_i-p}$, where $p = \min_i p_i$ is the unique minimum and $\nu = \arg\min_i p_i$. This result was derived under the (impractical) assumption that the size of the finite field over which the coding scheme operates is infinitely large.

Related to this result, Theorem~\ref{thm:DenseCodesRegularBernoulliAverageNon-IdenticalGeneral} or Theorem~\ref{thm:DenseCodesRegularBernoulliAverageNon-Identical} indicate that the average coding delay of dense codes over line networks with traffics as above, but with arbitrary or unequal parameters,\footnote{The special case of traffic parameters with a unique minimum can fall into each category of arbitrary or unequal traffic parameters. For example, aside from the uniqueness of the parameter with the minimum value, some other parameters might be equal, and hence such a case does not belong to the category of unequal parameters but it does belong to the category of arbitrary parameters.} is upper bounded by $\frac{1}{p}(k+(1+o(1))(\sqrt{kL\log(kL)}))$, or $\frac{1}{p}(k+(1+o(1))(L f(k)\log(kL)))$, respectively, where $f(k)$ goes to infinity sufficiently slow (see Theorem~\ref{thm:DenseCodesRegularBernoulliAverageNon-Identical}), as $k$ goes to infinity. It is important to note that both Theorems~\ref{thm:DenseCodesRegularBernoulliAverageNon-IdenticalGeneral} and~\ref{thm:DenseCodesRegularBernoulliAverageNon-Identical} do not have the limiting assumption of the result of~\cite{DDHE:2009} regarding the size of the finite field. The bounds of Theorems~\ref{thm:DenseCodesRegularBernoulliAverageNon-IdenticalGeneral} and~\ref{thm:DenseCodesRegularBernoulliAverageNon-Identical} are larger than that of~\cite{DDHE:2009}, which is expected, since the former, unlike the latter, are derived based on the practical assumption of operating over a finite field of size as small as two.

The results of Theorems~\ref{thm:DenseCodesRegularBernoulliActualNon-IdenticalGeneral} and~\ref{thm:DenseCodesRegularBernoulliActualNon-Identical} indicate that for both traffics with arbitrary or unequal parameters, the coding delay is upper bounded by $\frac{1}{p}(k+(1+o(1))(k^2 L\log(kL))^{\frac{1}{3}})$. This is while, in~\cite{DDHE:2009}, the coding delay is upper bounded by $\frac{1}{p}(k+O(k^{\frac{3}{4}}))$. This bound is looser than the bound in Theorem~\ref{thm:DenseCodesRegularBernoulliActualNon-IdenticalGeneral}, or the one in Theorem~\ref{thm:DenseCodesRegularBernoulliActualNon-Identical}, although it is derived under the same limiting assumption as the one used in~\cite{DDHE:2009} for the average coding delay (i.e., the size of the finite field being infinitely large). Such an assumption makes the bound appear smaller than what it would be at the absence of the assumption. This demonstrates the strength of the bounding technique used in this work.

By combining Theorems~\ref{thm:DenseCodesRegularBernoulliActualNon-IdenticalGeneral} and~\ref{thm:DenseCodesRegularBernoulliAverageNon-IdenticalGeneral}, or Theorems~\ref{thm:DenseCodesRegularBernoulliActualNon-Identical} and~\ref{thm:DenseCodesRegularBernoulliAverageNon-Identical}, it can be seen that the coding delay might be much larger than the average coding delay. This highlights the fact that the analysis of the average coding delay does not provide a complete picture of the speed of convergence of dense codes to the capacity of line networks.

\begin{table*}
  \caption{Comparison of Chunked Codes over Line Networks with Various Traffics}
  \hspace{-0.11in}
    \begin{tabular}{|p{1.275cm}|p{1cm}|@{}c@{}|@{}c@{}|@{}c@{}|@{}c@{}|}
    \hline
    \multirow{2}{*}{\hspace{.25 cm}\vspace{-.1cm}Traffic} & \multirow{2}{*}{\hspace{-.275 cm}$\vspace{-.1cm}\begin{array}{c} \text{Success} \\ \text{Parameters} \end{array}$} & $\begin{array}{c} \text{Overhead } \text{(}\eta\text{)} \\ \text{and} \end{array}$ & \multirow{2}[4]{*}{\vspace{.1cm} $\begin{array}{c} \text{Size of Chunks} \\ \text{(}\alpha\text{)}\end{array}$} & \multirow{2}[4]{*}{\vspace{.1cm}$w$} & \multirow{2}[4]{*}{\vspace{.1cm}Comments} \\
      &   & \hspace{.1cm}Average Overhead ($\bar{\eta}$)&   &   &  \\
    \hline
    $\begin{array}{c}\hspace{-.3cm} \text{Arbitrary} \\ \hspace{-.3cm}  \text{Deterministic}\end{array}$ & \hspace{.35cm} -  & $\eta=\bar{\eta}=O\left(kL\left(\frac{1}{\alpha} \log\frac{kL}{\epsilon}\right)^{\frac{1}{3}}\right)$ & $\omega\left({L^3\log\frac{kL}{\epsilon}}\right)$ &  -  & \multirow{5}[10]{*}{\vspace{-0.5cm}$\begin{array}{c} f(k)=\left\{o\left(\frac{k}{L\log\frac{kL}{\epsilon}}\right), \omega(1)\right\} \\ m= \frac{kw}{\alpha}\log\left(\frac{kLw}{\alpha\epsilon}\right) \\ \delta = \min\left\{\frac{\gamma_e}{p},1\right\} \\ \gamma_{e_i}=|p_i-p_{i-1}|\\ \gamma_e=\min_{1<i\leq L}\gamma_{e_i}\\ p=\min_{1\leq i\leq L}p_i \end{array}$} \\
\cline{1-5}  \multirow{4}{*}{\vspace{-0.75cm}$\begin{array}{c}\hspace{-.335cm}\text{Deterministic} \\ \hspace{-.335cm}\text{Regular} \\ \hspace{-.335cm}\text{Transmissions} \\ \hspace{-.335cm}\text{and} \\ \hspace{-.335cm}\text{Bernoulli} \\ \hspace{-.335cm}\text{Losses} \end{array}$} & \multirow{2}{*}{\vspace{-.25cm}Arbitrary} & $\eta = \frac{1}{p}\left((1+o(1))\left(\frac{kL}{w}+k^{\frac{1}{2}}m^{\frac{1}{2}}+m\right)\right)$ & \multirow{2}[4]{*}{\vspace{-.15cm}$\omega\left({L\log\frac{kL}{\epsilon}}\right)$} & $\left(\frac{\alpha L^2}{\log\frac{kL}{\epsilon}}\right)^{\frac{1}{3}}$ &  \\
\cline{3-3}\cline{5-5}      &   & $\bar{\eta}=\frac{1}{p}\left((1+o(1))\left(\frac{kL}{w}+m\right)\right)$ &   & $\left(\frac{\alpha L}{\log\frac{kL}{\epsilon}}\right)^{\frac{1}{2}}$ & \\
\cline{2-5}      & \multirow{2}{*}{\vspace{-.25cm}Unequal} & $\eta = \frac{1}{p}\left((1+o(1))\left(\frac{kL}{w}+k^{\frac{1}{2}}m^{\frac{1}{2}}\right)\right)$ & $\omega\left(\frac{L}{\delta} \log\frac{kL}{\epsilon}\right)$ & $\left(\frac{\alpha L^2}{\log\frac{kL}{\epsilon}}\right)^{\frac{1}{3}}$ & \\
\cline{3-5}      &   & $\bar{\eta}=\frac{1}{p}\left((1+o(1))\left(\frac{kL}{w}\right)\right)$ & $\omega\left(f(k)\cdot\frac{L}{\delta} \log\frac{kL}{\epsilon}\right)$ & $\frac{1}{f(k)}\left(\frac{\alpha}{\log\frac{kL}{\epsilon}}\right)$ & \\
    \hline
    \end{tabular}
  \label{tab:TableI}
\end{table*}

\subsection{Chunked Codes}\label{subsec:DiscussionChunkedCodes}
Table~\ref{tab:TableI} shows the upper bounds\footnote{With a slight abuse of language, we refer to the ``upper bound'' on the overhead or the average overhead as the ``overhead'' or the ``average overhead.''} (w.p. of failure b.a.b. $\epsilon$) on the overhead and the average overhead (i.e., the difference between the coding delay or the average coding delay and the capacity) of CC over various traffics for different ranges of the size of the chunks based on the results in Section~\ref{sec:BernoulliLossRegularTraffic} and those in~\cite{HBJ:2011}.\footnote{The results of Section~\ref{subsubsec:CCCACH} and those of Section~\ref{subsubsec:CCCAPP} were stated in terms of $q$ and $\alpha$, respectively. In this section, for the ease of comparison, the former results are also restated in terms of $\alpha$ by replacing $q$~with~$k/\alpha$.} The traffics under consideration are: arbitrary deterministic traffics, and traffics with deterministic regular transmissions and Bernoulli losses. We refer to the latter traffics as the \emph{probabilistic traffics} for simplification.\footnote{In the case of arbitrary deterministic traffics, the capacity is equal to $k$, and in the case of probabilistic traffics with parameters $\{p_i\}_{1\leq i\leq L}$, the capacity is equal to $k/p$, where $p=\min_{1\leq i\leq L}p_i$.} The probabilistic traffics are categorized into two sub-categories: traffics with arbitrary parameters and traffics with unequal parameters. We say that a code is ``capacity-achieving'' (c.-a.) if the ratio of the overhead to the capacity goes to $0$, as $k$ goes to infinity. Similarly, a code is ``capacity-achieving on average'' (c.-a.a.) if the ratio of the average overhead to the capacity goes to $0$, as $k$ goes to infinity. In Table~\ref{tab:TableI}, the upper (or the lower) row in front of each case of traffic parameters corresponds to a c.-a. (or a c.-a.a.) scenario.

In the table, one can see that, for each category of traffics, the size of the chunks ($\alpha$) has to be sufficiently large so that CC are c.-a. or c.-a.a.. For arbitrary deterministic traffics, the lower bound on $\alpha$ is super-logarithmic in $k$, i.e., $\omega(\log k)$, and super-log-cubic in $L$, i.e., $\omega({L^3\log L})$. For the probabilistic traffics with arbitrary or unequal parameters, the lower bound on $\alpha$ has a similar growth rate with $k$, but a smaller (super-log-linear) growth rate with $L$, i.e., $\omega({L\log L})$. The coding cost of CC (i.e., the number of the coding (packet) operations per message packet), is, on the other hand, linear in $\alpha$. Thus, CC can perform as fast over both the arbitrary deterministic traffics and the probabilistic traffics, but with a lower coding cost (smaller chunks) in the latter case compared to the former.

\begin{table*}
\caption{Comparison of Chunked Codes with Precoding over Line Networks with Various Traffics}
\hspace{-.11in}
    \begin{tabular}{|p{1.275cm}|p{1cm}|@{}c@{}|c|c|@{}c@{}|}
    \hline
    \multirow{2}{*}{\hspace{.25 cm}\vspace{-.1cm}Traffic} & \multirow{2}{*}{\hspace{-.275 cm}$\vspace{-.1cm}\begin{array}{c} \text{Success} \\ \text{Parameters} \end{array}$} & $\begin{array}{c} \text{Overhead } \text{(}\eta\text{)} \\ \text{and} \end{array}$ & \multicolumn{2}{c|}{\multirow{2}{*}{\vspace{-.1cm}$\begin{array}{c} \text{Size of Chunks} \\ \text{(}\alpha\text{)}\end{array}$}} & \multirow{2}{*}{\vspace{-.1cm}Comments} \\      &   & \hspace{.1cm}Average Overhead ($\bar{\eta}$) & \multicolumn{2}{c|}{} &  \\
    \hline
    $\begin{array}{c}\hspace{-.3cm} \text{Arbitrary} \\ \hspace{-.3cm}  \text{Deterministic}\end{array}$ & \hspace{.35cm} - & \hspace{.1cm}$\eta=\bar{\eta}=\gamma_o k$ & $\Omega\left(\frac{L^3}{\gamma^3_c}\log\frac{L}{\gamma_b\gamma_c}\right)$ & \multicolumn{1}{r|}{\multirow{5}[10]{*}{\vspace{-.25cm}$o\left(\sqrt{\frac{\gamma^2_a\gamma^2_b k}{\log\frac{1}{\epsilon}}}\right)$}} & \multirow{5}[10]{*}{\vspace{-.1cm}$\begin{array}{c} 0<\gamma_a,\gamma_b,\gamma_c<1 \\ \{\gamma_a,\gamma_b,\gamma_c\}=O(1) \\ \gamma_o=\gamma_c+(1+\gamma_c)\gamma'_o \\ \gamma'_o = (1+\gamma_a)\gamma_b+O(\gamma^2_b) \\  \gamma_{e_i}=|p_i-p_{i-1}| \\ \gamma_e=\min_{1<i\leq L}\gamma_{e_i} \\ p=\min_{1\leq i\leq L}p_i \end{array}$} \\
\cline{1-2}\cline{3-3}\cline{4-4}    \multirow{4}{*}{$\begin{array}{c}\hspace{-.335cm}\text{Deterministic} \\ \hspace{-.335cm}\text{Regular} \\ \hspace{-.335cm}\text{Transmissions} \\ \hspace{-.335cm}\text{and} \\ \hspace{-.335cm}\text{Bernoulli} \\ \hspace{-.335cm}\text{Losses} \end{array}$} & \multirow{2}[4]{*}{Arbitrary} &  {$\eta=\gamma_o\frac{k}{p}$} & $\Omega\left(\left\{\left(\frac{L}{\gamma^3_c}\log\frac{L}{\gamma_b\gamma_c}\right),\left(L^4\log\frac{L}{\gamma_b}\right)\right\}\right)^{\textcolor[rgb]{1.00,1.00,1.00}{\frac{1}{2}}}_{\textcolor[rgb]{1.00,1.00,1.00}{\frac{1}{2}}}$ &   &  \\ \cline{3-3}
\cline{4-4}      &   &  $\bar{\eta}=\gamma_o\frac{k}{p}$ & $\Omega\left(\frac{L}{\gamma_c}\log\frac{L}{\gamma_b\gamma_c}\right)^{\textcolor[rgb]{1.00,1.00,1.00}{\frac{1}{2}}}_{\textcolor[rgb]{1.00,1.00,1.00}{\frac{1}{2}}}$ &   &  \\ \cline{3-3}
\cline{2-2}\cline{4-4}      & \multirow{2}[4]{*}{Unequal} &  $\eta=\gamma_o\frac{k}{p}$ & $\Omega\left(\left\{\left(\frac{L}{\gamma^3_c}\log\frac{L}{\gamma_b\gamma_c}\right),\left(\frac{L}{\gamma^3_e}\log\frac{L}{\gamma_b\gamma_e}\right)\right\}\right)^{\textcolor[rgb]{1.00,1.00,1.00}{\frac{1}{2}}}_{\textcolor[rgb]{1.00,1.00,1.00}{\frac{1}{2}}}$ &   &  \\ \cline{3-3}
\cline{4-4}      &   & $\bar{\eta}=\gamma_o\frac{k}{p}$  & $\Omega\left(\frac{L}{\gamma^2_e\gamma_c}\log\frac{L}{\gamma_b\gamma_c}\right)^{\textcolor[rgb]{1.00,1.00,1.00}{\frac{1}{2}}}_{\textcolor[rgb]{1.00,1.00,1.00}{\frac{1}{2}}}$ &   &  \\ \cline{3-3}
    \hline
    \end{tabular}
  \label{tab:TableII}
\end{table*}

Moreover, as it can be seen in Table~\ref{tab:TableI}, for both arbitrary deterministic and probabilistic traffics (in each case of arbitrary or unequal traffic parameters), the overhead grows sub-log-linearly with $k$, i.e., $O(k\log^{\frac{1}{3}} k)$, and decays sub-linearly with $\alpha$, i.e., $O(1/\alpha^{\frac{1}{3}})$. However, for arbitrary deterministic traffic, the overhead grows with $O(L\log^{\frac{1}{3}}L)$, and for the probabilistic traffics, it only grows with $O(L^{\frac{1}{3}}\log^{\frac{1}{3}} L)$. This implies a faster speed of convergence to the capacity in the latter case compared to the former. Similar comparison results can also be observed in terms of the average overhead, except that in the case of unequal traffic parameters, the average overhead decays linearly with $\alpha$, i.e., $O(1/\alpha)$, but grows poly-log-linearly with $k$, i.e., $O(k\log^2 k)$, for the choice of $f(k)=O(\log k)$, and log-linearly with $L$, i.e., $O(L\log L)$.

Table~\ref{tab:TableII} shows the results for CC with precoding (CCP) in the scenarios similar to those considered in Table~\ref{tab:TableI}, where the precode is a capacity-achieving erasure code of dimension $k$ and rate $1-(1+\gamma_a)\gamma_b$. In particular, one can see that CCP are ``capacity-approaching'' or ``capacity-approaching on average'' with an arbitrary small ``non-zero constant'' gap $\gamma_o$ (i.e., the ratio of the overhead or the average overhead to the capacity goes to $\gamma_o$, as $k$ goes to infinity) if $\alpha$ is sufficiently large. For simplifying the terminology, we drop the term ``with a non-zero constant gap.'' The upper (or the lower) row in front of each case of traffic parameters corresponds to a capacity-approaching (or a capacity-approaching on average) scenario. For arbitrary deterministic traffics, the lower bound on $\alpha$ is constant in $k$, and log-cubic in $L$, i.e., $O(L^3\log L)$. For the probabilistic traffics with arbitrary or unequal parameters, the lower bound on $\alpha$ is also constant in $k$, but has a smaller (log-linear) growth rate with $L$, i.e., $O(L\log L)$. Thus, in the case of CCP, one can make a conclusion similar to the one made in the case of stand-alone CC, with respect to the arbitrary deterministic and the probabilistic traffics.

\bibliographystyle{IEEEtran}
\bibliography{RefsII}

\end{document}